\documentclass[nocopyrightspace,preprint]{sigplanconf}
\usepackage{stmaryrd,bcprulesmhfix,tensor}
\bcprulessavespacetrue
\suppressrulenamesfalse
\usepackage{mathpartir}
\usepackage{upgreek}
\usepackage{proof}
\usepackage[all]{xy}
\usepackage{graphicx}
\usepackage{amsmath}
\usepackage{lineno}
\usepackage{xypic}
\usepackage{mathtools}
\usepackage{url,amssymb}
\renewcommand{\phi}{\varphi}
\usepackage{amsthm}
\usepackage{xspace}
\usepackage{comment}
\RequirePackage{txfonts}

\usepackage[compact]{titlesec}
\titlespacing{\section}{0pt}{1ex}{1ex}
\titlespacing{\subsection}{0pt}{1ex}{0ex}
\titlespacing{\subsubsection}{0pt}{0.5ex}{0ex}
\usepackage{color}

\definecolor{darkred}{rgb}{0.5,0,0}
\definecolor{darkgreen}{rgb}{0,0.5,0}
\definecolor{darkblue}{rgb}{0,0,0.5}

\usepackage{ifpdf}

\usepackage{paralist}

\makeatletter
\newcommand*{\bdiv}{%
  \nonscript\mskip-\medmuskip\mkern5mu%
  \mathbin{\operator@font div}\penalty900\mkern5mu%
  \nonscript\mskip-\medmuskip
}
\makeatother
\newcommand{\keywd}[1]{\mathtt{#1}}

\newcommand{\myread}[1]{!{#1}}
\newcommand{\wcpo}[1]{$\omega$-cpo}
\newcommand{\wcpos}[1]{$\omega$-cpos}
\newcommand{\myref}[1]{\keywd{ref}(#1)}

\newcommand{\intt}[1]{\textit{int}(#1)}
\newcommand{\myeffto}[3]{\xrightarrow[#2]{#1\,\mid\, #3}}

\newcommand{\mnond}{\texttt{?}}

\newcommand{\funn}[1]{\textit{fun}(#1)}

\newcommand{\mylet}{\keywd{let}}
\newcommand{\FV}{\textit{FV}}

\newcommand{\partfun}{\rightharpoondown}

\newcommand{\mif}{\keywd{if}}
\newcommand{\mthen}{\keywd{then}}
\newcommand{\melse}{\keywd{else}}

\newcommand{\mequals}{\keywd{=}}

\newcommand{\mtrue}{\keywd{true}}
\newcommand{\mfalse}{\keywd{false}}

\newcommand{\ints}{\mathbb{Z}}

\newcommand{\inttype}{\keywd{int}}

\newcommand{\booltype}{\keywd{bool}}
\newcommand{\unittype}{\keywd{unit}}
\newcommand{\unitval}{\keywd{()}}

\newcommand{\labs}{\mathbb{L}}
\newcommand{\dom}[1]{\mathrm{dom}({#1})}

\newtheorem{assumption}{Assumption}

\newcommand{\pause}{\vspace{1.5ex}}

\newcommand{\ba}{\begin{array}}
\newcommand{\ea}{\end{array}}

\newcommand{\idd}{\textit{id}}

\newcommand{\squelch}[1]{}

\newcommand\orth[2]{\ensuremath{#1\,\bot\, #2}\xspace}

\newcommand{\vfix}[3]{\keywd{rec}\:{#1}\:{#2} = {#3}}
\newcommand{\letin}[2]{\keywd{let}\:{#1}\!=\!{#2}\:\keywd{in}\:}

\newcommand{\opletrec}[3]{\keywd{let\ rec}}

\newcommand{\assign}[2]{{#1}:={#2}}

\newcommand{\rdsin}[1]{\mathrm{rds}({#1})}

\newcommand{\wrsin}[1]{\mathrm{wrs}({#1})}

\newcommand{\myif}[3]{\keywd{if}\ #1\ \keywd{then}\ #2\
  \keywd{else}\ #3}
\newcommand{\myatomic}[1]{\keywd{atomic}(#1)}  
\newcommand{\cas}[3]{\keywd{cas}( #1, #2, #3)}

\newcommand{\mypar}[2]{#1 \| #2}
\newcommand{\semparallel}{~|~}

\newcommand{\eff}{\varepsilon}
\newcommand{\reads}{\mathrm{rds}}
\newcommand{\rds}{\reads}
\newcommand{\writes}{\mathrm{wrs}}
\newcommand{\wrs}{\writes}
\newcommand{\concs}[1]{\mathrm{cos}(#1)}
\newcommand{\locs}[1]{\mathrm{locs}(#1)}

\renewcommand{\topfraction}{.8}

\newcommand{\fullonly}[1]{}

\newlength{\lruleZeroname}
               {\vspace{1em}{\bf Example.}
               }%
               {
               \qed%
               \vspace{1em}%
               }

\renewcommand{\P}{\ensuremath{\mathcal{P}}}

\newcommand{\Locs}{\ensuremath{\labs}}
\newcommand{\new}{\ensuremath{\mathit{new}}}

\newcommand{\Stores}{\ensuremath{\mathbb{H}}}
\newcommand{\sem}[1]{\ensuremath{\llbracket {#1} \rrbracket}}
\newcommand{\semC}[1]{\ensuremath{\llbracket {#1} \rrbracket}}
\newcommand{\semV}[1]{\ensuremath{\llceil {#1} \rrceil
}}

\newcommand{\rEff}[1]{\ensuremath{\mathit{rd}_{#1}}}
\newcommand{\cEff}[1]{\ensuremath{\mathit{co}_{#1}}}
\newcommand{\wEff}[1]{\ensuremath{\mathit{wr}_{#1}}}

\newtheorem{theorem}{Theorem}[section]
\newtheorem{lemma}[theorem]{Lemma}
\newtheorem{definition}[theorem]{Definition}
\newtheorem{example}[theorem]{Example}
\newcounter{Examplecount}
\newcommand\im{\mathrm{Img}}

\newcommand{\loc}{\mathfrak{l}}

\newcommand{\locListOdd}{\ensuremath{\mathfrak{listodd}}\xspace}
\newcommand{\locListEven}{\ensuremath{\mathfrak{listeven}}\xspace}

\newcommand{\locInt}{\ensuremath{\mathfrak{int}}\xspace}
\newcommand{\locHi}{\ensuremath{\mathfrak{snd}}\xspace}
\newcommand{\locLo}{\ensuremath{\mathfrak{fst}}\xspace}

\newcommand{\locMSQ}{\ensuremath{\mathfrak{msq}}\xspace}

\newcommand{\cloc}{\ensuremath{X}\xspace}

\newcommand\E{\ensuremath{\,E\,}\xspace}

\newcommand\w{\ensuremath{\mathsf{w}}\xspace}

\newcommand\q{\ensuremath{\mathsf{q}}\xspace}

\newcommand\heap{\ensuremath{\mathsf{h}}\xspace}
\newcommand\hinit{\ensuremath{\mathsf{h}_{\mathit{init}}}\xspace}
\newcommand\h{\heap}
\renewcommand\k{\ensuremath{\mathsf{k}}\xspace}
\newcommand\heapp{\ensuremath{\mathsf{q}}}
\newcommand{\Values}{\mathbb{V}}
\newcommand{\ValuesB}{\mathbb{VB}}

\newcommand\vval{\ensuremath{a}\xspace}

\newcommand\cval{\ensuremath{\mathsf{c}}\xspace}

\newcommand{\myety}[4]{{#1}\mathrel{\&}{#2} \mid #3 \mid #4}

\newcommand{\valty}[1]{#1}

\newcommand\inR[2]{\ensuremath{#1:#2}}

\newcommand\rloc[3]{\ensuremath{#1 \stackrel{#3}{\sim} #2}}
\newcommand\rrloc[3]{\ensuremath{#1 \stackrel{#3}{=} #2}}
\newcommand\gloc[3]{\ensuremath{#1 \xrightarrow{#3}#2}}

\newcommand\ie{\emph{i.e.}\xspace}

\newcommand\etal{\emph{et al.}\xspace}

\newcommand{\blue}[1]{\textcolor{blue}{#1} }

\newcommand{\tup}[1]{(#1)}

\usepackage{paralist}


\title{Effect-Dependent Transformations for Concurrent Programs}
\authorinfo{Nick Benton}{Microsoft Research, Cambridge, UK}{nick@microsoft.com}
\authorinfo{Martin Hofmann}{LMU, Munich, Germany}{hofmann@ifi.lmu.de}
\authorinfo{Vivek Nigam}{UFPB, Jo\~ao Pessoa, Brazil}{vivek.nigam@gmail.com}

\begin{document}
\maketitle
\begin{abstract}  
We describe a denotational semantics for an abstract effect
system for a higher-order, shared-variable concurrent programming
language. We prove the soundness of a number of general effect-based
program equivalences, including a parallelization equation that
specifies sufficient conditions for replacing sequential composition
with parallel composition. Effect annotations are relative to abstract 
locations specified by contracts rather than physical footprints allowing us 
in particular to show the soundness of
some transformations involving fine-grained concurrent data structures, such as
Michael-Scott queues, that allow concurrent access to different parts
of mutable data structures.

Our semantics is based on refining a trace-based semantics for
first-order programs due to Brookes. By moving from concrete to
abstract locations, and adding type refinements that capture the
possible side-effects of both expressions and their concurrent
environments, we are able to validate many equivalences that do not
hold in an unrefined model.  The meanings of types are expressed using
a game-based logical relation over sets of traces. Two programs $e_1$
and $e_2$ are logically related if one is able to solve a two-player
game: for any trace with result value $v_1$ in the semantics of $e_1$
(challenge) that the player presents, the opponent can present an
(response) equivalent trace in the semantics of $e_2$ with a logically related result value $v_2$.

\end{abstract}
\squelch{
}
\noindent
\section{Introduction}
\label{sec:intro}
Type-and-effect systems refine conventional types with extra
static information capturing a safe upper bound on the possible
side-effects of expression evaluation. Since their introduction by
Gifford and Lucassen \cite{DBLP:conf/lfp/GiffordL86},
effect systems have been used for many purposes, including
region-based memory management \cite{birkedaltoftevejlstrup}, tracking
exceptions \cite{pessauxleroy,bentonbuchlovsky}, communication
behaviour \cite{amtoftnielsons} and atomicity
\cite{flanaganqadeerpldi03} for concurrent programs, and information
flow \cite{brobergsands:flowlocks}. 

A major reason for
tracking effects is to justify program
transformations, most obviously in optimizing compilation \cite{DBLP:conf/icfp/BentonKR98}. For example,
one may remove computations whose results are
unused, \emph{provided} that they are sufficiently pure, or commute two
state-manipulating computations, \emph{provided} that the locations they may
read and write are suitably disjoint. Several groups 
have recently studied the semantics of effect systems, with a focus on
formally justifying such effect-dependent equational reasoning \cite{DBLP:conf/popl/KammarP12,DBLP:conf/aplas/BentonKHB06,DBLP:dblp_conf/popl/Benton0N14,birkedal,DBLP:conf/icfp/ThamsborgB11}. A common approach, which we follow here, is to
interpret effect-refined types using a
logical relation over the (denotational or operational) semantics of
the unrefined (or untyped) language, simultaneously identifying both the
subset of computations that have a particular effect type and a
coarser notion of equivalence (or approximation) on that subset. Such
a semantic approach decouples the meaning of effect-refined types from
particular syntactic rules: one may establish that a term
has a type using various more or less approximate inference systems, or
by detailed semantic reasoning. 

For sequential computations with global state,  denotational
models already provide significant
abstraction. For example, the denotations of \verb|skip| and
\verb|X++;X--| are typically equal, so it is immediate that the
second is semantically pure. More generally, 
the meaning of a
judgement $\Gamma\vdash e:\tau \& \eff$ guarantees that the result of
evaluating $e$ will be of type $\tau$ with side-effects at most $\eff$,
under assumptions $\Gamma$ (a `rely' condition), on the behaviour of
$e$'s free variables. The possible interaction points
between $e$ and its environment are restricted to initial states and
parameter values, and final states and results, of $e$ itself and its
explicitly-listed free variables. Furthermore, all those interaction
points are visible in the term and are governed by specific
annotations appearing in the typing judgement.

For shared-variable concurrency, there are many more possible
interactions. An expression's environment now also includes anything
that may be running concurrently and, moreover, atomic steps of $e$
and its concurrent environment may be arbitrarily interleaved, so it
is no longer sufficient to just consider initial and final states. A
priori, this leads to far fewer equations between programs. For
example, \verb|X++;X--| may be distinguished from \verb|skip| by being
run concurrently with a command that reads or writes \verb|X|. But few
programs do anything useful in the presence of unconstrained
interference, so we need ways to describe and control
it. 
Fine-grained, optimistic algorithms, which rely on
custom protocols being followed by multiple threads with concurrent
access to a shared data structure, can significantly
outperform ones based on coarse-grained locking, but are notoriously
challenging to write and verify.

There is a huge literature on shared-variable
concurrency, from type systems ensuring race-freedom of programs with
locks \cite{abadi:typessafelock} to sophisticated semantic models for
reasoning about refinement of fine-grained concurrent datastructures
\cite{dreyer}. This paper explores effect types as a
straightforward, lightweight interface language for modular reasoning
about equivalence and refinement, e.g. for safely transforming
sequential composition into parallelism.  We show how the semantics of
a simple effect system scales smoothly to the concurrent setting,
allowing us to control interference and prove non-trivial
equivalences, extending (somewhat to our surprise) to the correctness
of some fine-grained algorithms.

We build on a trace semantics for concurrent programs, due to Brookes
\cite{brookes96ic}, which explicitly describes possible interference
by the environment. We extend Brookes's semantics to a higher-order
language and then refine it by a semantically-formulated effect system
that separately tracks: (1) the store effects of an expression during
evaluation; (2) the assumed effects of transitions by the environment;
and (3) the overall end-to-end effect.  
Rather than tracking effects at
the level of individual concrete heap cells, we view the heap as a set
of abstract data structures, each of which may span several locations,
or parts of locations \cite{DBLP:dblp_conf/popl/Benton0N14}. Each abstract location has its own notion of
equality, and its own notion of legal mutation. Write effects, for
example, need only be flagged when the equivalence class of an
abstract location may change. Both typing and refinement judgements may be established by a combination of generic type-based rules and semantic reasoning in the model.






\label{examples}
We begin with some motivating examples. 

\paragraph{Equivalence modulo non-interference:}
Our semantics justifies the following equation \emph{at} the  effect type 
$\unittype\ \&\ \{\cEff{\cloc}\} \mid \eff \mid \eff \cup \{\rEff{\cloc}, \wEff{\cloc}\}$:
\[\begin{array}{l}
 (\assign{\cloc}{\myread{\cloc} + 1}; \assign{\cloc}{\myread{\cloc}+ 1})   = 
 (\assign{\cloc}{\myread{\cloc} + 2}) 
\end{array}
\]
This says that the two commands are equivalent with return type \texttt{unit}, exhibit the effect $\cEff{\cloc}$, signifying concurrent or `chaotic' access to $\cloc$ along the way, and have an overall end-to-end effect of $\eff$ plus reading and writing $\cloc$, \emph{provided} that the effect, $\eff$, of the concurrent environment does not involve $\cloc$.
\nopagebreak

\paragraph{Overlapping References:}
Let $p$,$p^{-1}$ implement a bijection $\ints \to \ints\times\ints$, and consider the following functions:
\[ 
\begin{array}{l}
      \mathsf{readFst}~() = p(\myread{\cloc}).1 \\
       \mathsf{readSnd}~() = p(\myread{\cloc}).2 \\
\mathsf{wrtFst}~ n =   \keywd{let \ rec}~ \mathsf{try}~() = (\keywd{let}~ m = \myread{\cloc} ~ \keywd{in}~ \letin{(x,y)}{p(m)}~ 
\\ \quad \letin{m'}{p^{-1}(n,y)} ~ \keywd{if}~\cas{\cloc}{m}{m'}~\keywd{then}~()~ \keywd{else}~\mathsf{try}~ ()) \\
  ~\keywd{in}~\mathsf{try}~() \\
  \mathsf{wrtSnd}~ n =   \keywd{let \ rec}~ \mathsf{try}~() =  \keywd{let}~ m = \myread{\cloc} ~ \keywd{in}~
 \letin{(x,y)}{p(m)}~ \\ \quad 
 \letin{m'}{p^{-1}(x,n)} ~ \keywd{if}~\cas{\cloc}{m}{m'}~\keywd{then}~()~ \keywd{else}~\mathsf{try}~ () \\  
 ~\keywd{in}~\mathsf{try}~()
\end{array}
\]
which multiplex two abstract integer references onto a single concrete one. Note that the write functions, $\mathsf{wrtFst}$ and $\mathsf{wrtSnd}$, use compare-and-swap,  $\keywd{cas}$, to atomically update the value of the reference. 


Our generic rules then say that a program, $e_1$, that only reads and/or writes one abstract  reference can be commuted, or executed in parallel, with another program, $e_2$, that only reads and/or writes into a different reference. This lets one use types to, say, justify parallelizing a call to $\mathsf{wrtFst}$ followed by one to $\mathsf{wrtSnd}$, even though they read and write the same concrete location, which looks like a race.

\paragraph{Version numbers:}
\newcommand{\clocver}{{\cloc}_{\textrm{ver}}}
\newcommand{\clocval}{{\cloc}_{\textrm{val}}}
One can isolate a transaction that reads and then writes a piece of state simply by enclosing the whole thing in $\myatomic{\cdot}$. A more concurrent alternative adds a monotonic version number to the data. A transaction then works on a private copy, only committing its changes back (and incrementing the version) if the current version number is the same as that of the original copy. We can define an abstract integer reference $\mathfrak{X}$ in terms of two concrete ones, $\clocver$ and $\clocval$, governed by a specification that says  $\myread{\clocval}$ may only change when $\myread{\clocver}$ increases. We define 
\[\begin{array}{l}
\mathsf{transact}~f = \keywd{let\ rec}~\mathsf{try}()= \keywd{let}~(val,ver)=\myatomic{(\myread{\clocval},\myread{\clocver})}\\
\ \keywd{in}~\keywd{let}~res = f(val)\ \keywd{in\ if\ atomic }(\keywd{if}\ \myread{\clocver}=ver\ \keywd{then}\\
\quad \clocver := ver+1;\ \clocval := res;\ \mtrue\ \keywd{else}\ \mfalse)\\
\ \keywd{then}\ \unitval\ \keywd{else}\ \mathsf{try}\unitval\\
\keywd{in}\ \mathsf{try}()
\end{array}
\]
Under the assumption that $f$ is a pure function (has effect type $\inttype\myeffto{\emptyset}{\eff}{\eff}\inttype$ for any $\eff$), we can show
\[
\mathsf{transact}~f = \myatomic{\clocval := f(\myread{\clocval}); \clocver := \myread{\clocver}+1}
\]
at type $\unittype \& \{\rEff{\mathfrak{X}},\wEff{\mathfrak{X}}\} \mid \eff \mid \eff\cup\{\rEff{\mathfrak{X}},\wEff{\mathfrak{X}}\}$ for any $\eff$ not including chaotic access, $\cEff{\mathfrak{X}}$, to $\mathfrak{X}$. The environment effect $\eff$ here \emph{may} include reading and writing $\mathfrak{X}$, so concurrent calls to $\mathsf{transact}$ are linearizable. 

\paragraph{Loop Parallelization:}
Our next example is inspired by a loop unrolling optimization~\cite{DBLP:conf/popl/TristanL10}. Assume given a linked list of integers pointed by $head$. Consider the following functions: 
\[ 
\begin{array}{ll}
  \mathsf{map}~f = &  \keywd{let \ rec }~\mathsf{applyf}~n = \\ 
 & \quad  n.ele := f(n.ele); \keywd{if}~n.next = null~\keywd{then}~\unitval\\
  &  \quad \keywd{else}~\mathsf{applyf}~(n.next)\\
  & \keywd{in}~\keywd{if}~!head = null~\keywd{then}~\unitval~\keywd{else}~\mathsf{applyf}~(!head)\\[2pt]
  \mathsf{map2Par}~f = & \keywd{let \ rec }~\mathsf{applyf2}~n = \\ 
  & \quad \blue{n.ele := f(n.ele) ~ {\underline{~\|~}} ~ {n.next.ele := f(n.next.ele)}}; \\
  & \quad \keywd{if}~n.next.next = null~\keywd{then}~\unitval\\
  & \quad \keywd{else}~\keywd{if}~n.next.next.next = null~\keywd{then} \\ 
  & \qquad \quad n.next.next.ele = f(n.next.next.ele)\\
  & \qquad \keywd{else}~ \mathsf{applyf2}~(n.next.next)\\
  & \keywd{in}~ \keywd{if}~!head = null~\keywd{then}~\unitval \\
  & \quad~ \keywd{else}~\keywd{if}~!head.next = null~ \keywd{then}\\
  & \qquad  !head.next.ele := f(!head.next.ele)\\
  & \quad ~ \keywd{else}~\mathsf{applyf2}~(!head) 
\end{array}
\]  
The function $\mathsf{map}$ simply applies a pure function $f$ to each element of the list, each element per iteration. The function $\mathsf{map2Par}$, on the other hand, applies $f$ to two consecutive elements of the list in parallel, potentially allowing one to exploit multiple cores. Our effect-based reasoning will soundly transform $\mathsf{map}$ into $\mathsf{map2Par}$ (under the assumption that the environment does not interfere with the list).  

\begin{figure}[t]
\begin{center}
 \includegraphics[width=8cm]{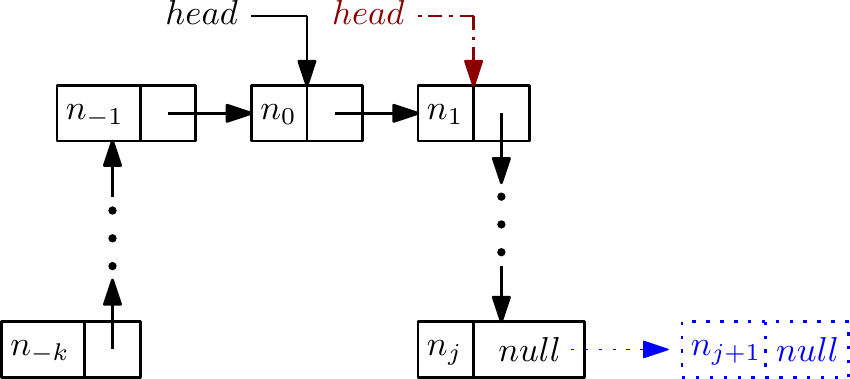} 
 \end{center}
 \caption{Illustration of a Michael-Scott Queue. The list resulting from the pointer to the element $n_0$ (the $head$ pointer with the continuous arrow in black) contains the list of elements $[n_1, \ldots, n_j]$. The enqueueing operation is illustrated by the dotted arrow and the box with the element $n_{j+1}$ (in blue), while the dequeueing operation is illustrated by the dot dashed head pointer (in red).}
 \label{fig:MSQ}
  \[
  \begin{array}{ll}
      \mathsf{dequeue}~ () = & \keywd{let\ rec}~ \mathsf{try}~() = \\&
   \quad \keywd{let}~ n_0 = \myread{head} ~\keywd{in}~
   \keywd{if}~\myread{n_0}.next = null~ \keywd{then}~ null~ 
   \\ & \qquad \keywd{else}~ \keywd{let}~ n_1 = \myread{n_0}.next ~\keywd{in}~
   \\ & \quad \qquad \keywd{if}~\cas{\myread{head}}{n_0}{n_1}~\keywd{then}~\myread{n_1}.ele~ \keywd{else}~ \mathsf{try}~() \\&
   \keywd{in}~\mathsf{try}~() \\[5pt]


   \mathsf{enqueue}(x) =& 
  \keywd{let\ rec}~\mathsf{try}~(p) = \\
 &\ \keywd{if}\ !p.next=null\ \keywd{then}\\
   & ~~~ \keywd{if}~ \keywd{atomic}( \keywd{if}~!p.next=null~ \keywd{then}~\\
   & \qquad !p.next:=\myref{x,null}; \mtrue~\keywd{else}~\mfalse)\\
  &~~~ \keywd{then}~ ()~ \keywd{else}~ \mathsf{try}~ (\myread{p}.next)\\ 
 &\ \keywd{else}\ \mathsf{try}\ (\myread{p}.next)\\
 &\keywd{in}~ \mathsf{try}~(\myread{head})
  \end{array}
 \] 
  \caption{Enqueue and Dequeue programs for a Michael-Scott Queue at location $head$.}
 \label{fig:deqenq}
 \vspace{-3mm}
 \end{figure}

 \paragraph{Michael-Scott Queue:} The Michael-Scott
 Queue~\cite{michael-scott} (MSQ) is a fine grained concurrent data
 structure, allowing threads to access and modify
 different parts of a queue safely and simultaneously. We present a version like that of Turon et al  \cite{dreyer}, which is an idealized version of the MSQ, without a tail pointer. 
 
 An MSQ maintains a pointer $head$ to a non-empty linked list as depicted in Figure~\ref{fig:MSQ}. The first node, the node containing the element $n_0$ in the figure, is not an element of the queue, but is a ``sentinel''. Hence the queue in the figure holds $[n_1, \ldots, n_j]$. 
 
 The enqueue and dequeue operations are defined in Figure~\ref{fig:deqenq} and illustrated in Figure~\ref{fig:MSQ}.
 Elements are dequeued from the beginning of the linked list, and enqueued at the end, which involves a traversal that is done without locking.
Once the end, $p$, of the linked list is found, the program atomically attempts to insert the new element. This is necessary because other programs may have enqueued elements to the end of the list, meaning that $p$ is no longer the end of the list. 
   
 The dequeue operation should move the $head$ pointer from the current
 sentinel, $n_0$, to the following element $n_1$. However, as other
 programs may also be attempting to dequeue an element, we use
 compare-and-swap to atomically update  the $head$ pointer if $head$
 still points to the same sentinel. Notice that the dequeued elements
 can still reach the sentinel of the queue. (In Figure~\ref{fig:MSQ},
 these are the nodes containing $n_{-k}, \ldots, n_{-1}$.)  This is
 necessary because there might be other (slower) threads that want to
 enqueue an element and are still searching for the end of the list
 by traversing the portion of the queue that has already been
 dequeued. 
 
 We prove  that the enqueue and dequeue of Figure~\ref{fig:deqenq} are equivalent to $\myatomic{\mathsf{enqueue}}$ and $\myatomic{\mathsf{dequeue}}$, their atomic versions 
which perform all operations in a single step, at a type that allows the environment to be concurrently reading and writing the queue.
  So the fine-grained MSQ behaves like a synchronized queue, as might also be implemented using locks.
\section{Syntax}
\label{sec:syntax}
In this section we define the syntax of a metalanguage for concurrent, stateful
computations and higher-order functions. 
Communication between parallel computations is via a shared heap
mapping dynamically allocated locations to structured values, which
include pointers. To keep the model simple, we do not allow functions
to be stored in the heap (no higher-order store).


\paragraph{Memory model}
We assume a countably infinite set $\labs$ of physical locations
$\cloc_1, \ldots, \cloc_n, \ldots$ and a set $\ValuesB$ of
``R-values'' that can be stored in those references including
integers, booleans, locations, and tuples of R-values, written $(v_1,
\ldots, v_n)$.  We assume that it is possible to tell of which form a
value is and to retrieve its components in case it is a tuple.  A heap $\h$, then, is a \emph{finite map} from $\labs$ to $\ValuesB$, written
$\{(\cloc_1, \cval_1),(\cloc_2, \cval_2), \ldots, (\cloc_n,
\cval_n)\}$, specifying that the value stored in location $\cloc_i$ is $\cval_i$. We write $\dom{\h}$ for the domain of $\h$ and
write $\h[\cloc{\mapsto} \cval]$ for the heap that agrees with $\h$
except that it gives the variable $\cloc$ the value $\cval$. The set
of heaps is denoted by $\Stores$.  We also assume that $\new(\heap,v)$
yields a pair $(\cloc,\heap')$ where $\cloc\in\Locs$ is a fresh
location and $\heap'\in\Stores$ is $\heap[\cloc{\mapsto}v]$. 

\paragraph{Syntax of expressions}
The syntax of untyped values and computations is:
\[
 \begin{array}{lcl}
v & ::= & x \mid (v_1,v_2) \mid v_r\mid c \mid
\vfix{f}{x}{t} \\
e &::=& v\mid \letin{x}{e_1}{e_2}\mid v_1\,v_2 \mid
\myif{v}{e_1}{e_2} \\&&
\mid \myread{v}\mid \assign{v_1}{v_2}\mid \myref{v}
 \mid \mypar{e_1}{e_2} \mid
\myatomic{e}
 \end{array}
\]
Here, $x$ ranges over variables, $v_r$ over R-values, and $c$
over built-in functions, which include arithmetic, testing
whether a value is an integer, function, pair or reference,
equality on simple values, etc. Each $c$ has a corresponding semantic
partial function $F_c$, so for example $F_+(n,n')=n+n'$ for integers $n,n'$.

The construct $\vfix{f}{x}{e}$ defines a recursive function with body
$e$ and recursive calls made via $f$; we use $\lambda x.e$ as
syntactic sugar in the case when 
$f$ is not free in $e$. Next, $\myread{v}$
(reading) returns the contents of location $v$, $\assign{v_1}{v_2}$
(writing) updates location $v_1$ with value $v_2$, and $\myref{v}$
(allocating) returns a fresh location initialized with $v$. The
metatheory is simplified by using ``let-normal form'', in which the
only elimination for computations is $\mylet$, though we sometimes nest
computations as shorthand for let-expanded versions in examples.
 
The construct $\mypar{e_1}{e_2}$ is evaluated by arbitrarily
interleaving evaluation steps of $e_1$ and $e_2$
until each has produced a value, say $v_1$ and $v_2$; the result is then
$(v_1,v_2)$. Assignment, dereferencing and allocation are atomic, but
evaluation of nested expressions is generally not.
To enforce atomicity, 
$\myatomic{e}$ evaluates an arbitrary $e$ in one step, without any environmental interference.
One can then define a (more realistic) compare-and-swap operation
$\cas{\cloc}{v_1}{v_2}$:
\[
\hspace{-2mm}
  \cas{\cloc}{v_1}{v_2} = \myatomic{\keywd{if}~\myread{\cloc} = v_1~\keywd{then}~\cloc := v_2; \mtrue~\keywd{else}~\mfalse}
\]
this atomically both checks if location $\cloc$ contains $v_1$ and, if so,
replaces it with $v_2$ and returns $\mtrue$; otherwise the location is unchanged and the returned value is $\mfalse$.

We define the free variables, $\FV(e)$, of a term, closed terms, and
the substitution $e[v/x]$ of $v$ for $x$ in $e$, in the usual
way. Locations may occur in terms, but the type system will
constrain their use.



\section{Denotational Model}\label{values}
We now sketch a denotational semantics for our metalanguage based on
Brookes' trace semantics \cite{brookes96ic}. 
Fuller details can be found in a technical report (attached), which in
particular establishes computational adequacy of the model with
respect to a small-step operational semantics using interleaving.

\subsection{Preliminaries}
A \emph{predomain} is an $\omega$-cpo, \ie, a partial order with
suprema of ascending chains.  A \emph{domain} is a predomain with a
least element, $\bot$.  Recall that $f:A\rightarrow A'$ is
\emph{continuous} if it is monotone $x\leq y \Rightarrow f(x)\leq
f(y)$ and preserves suprema of chains, \ie, $f(\sup_i x_i)=\sup_i
f(x_i)$. Any set is a predomain with the discrete order (flat
predomain). If $X$ is a set and $A$ a predomain then any
$f:X\rightarrow A$ is continuous. We denote a partial (continuous)
function from set (predomain) $A$ to set (predomain) $B$ by $f:A
\partfun B$.
If $A,B$ are predomains the cartesian product $A\times B$ and the set
of continuous functions $A{\rightarrow}B$ form themselves predomains
(with the obvious componentwise and pointwise orders) and make the
category of predomains cartesian closed. Likewise, the partial continuous functions $A{\partfun}B$ between predomains $A, B$ form a domain. 

If $P\subseteq A$ and $Q\subseteq B$ are subsets of predomains $A$ and
$B$ we define $P\times Q\subseteq A\times B$ and
$P{\rightarrow}Q\subseteq A{\rightarrow}B$ in the usual way. We may
write $f:P\rightarrow Q$ for $f\in P{\rightarrow}Q$.  

A subset $U\subseteq A$ is \emph{admissible} if whenever $(a_i)_i$ is
an ascending chain in $A$ such that $a_i\in U$ for all $i$, then
$\sup_i a_i\in U$, too. If $f:X\times A\rightarrow A$ is continuous
and $A$ is a domain then one defines $f^\ddagger(x)=\sup_if_x^i(\bot)$
with $f_x(a)=f(x,a)$. One has, $f(x,f^\ddagger(x))=f^\ddagger(x)$ and
if $U\subseteq A$ is admissible and contains $\bot$ and $f:X\times
U\rightarrow U$ then $f^\ddagger:X\rightarrow U$, too.  An element $d$
of a predomain $A$ is \emph{compact} if whenever $d\leq\sup_i a_i$
then $d\leq a_i$ for some $i$. E.g.\ in the domain of partial
functions from $\mathbb{N}$ to $\mathbb{N}$ the compact elements are
precisely the finite ones.  A continuous partial function $f:A\partfun
A$ is a \emph{retract} if $f(a)\leq a$ and $f(f(a))=f(a)$ hold for all
$a\in A$. In short: $f\leq\textit{id}_A$ and $f\circ f\leq f$. If, in
addition, $f$ has a finite image then $f$ is called a \emph{deflation}
\cite{Abramsky94domaintheory}. Note that if $f$ is a retract then $\dom f=\im(f)$
and if $a\in \im(f)$ then $a=f(a)$. We also note that if $a$ is in the
image of a deflation then $a$ is compact.
 
We define the usual state monad on predomains, by taking $ SA =
\Stores\partfun \Stores\times A $. 

\begin{definition}
Let $P$ be a subset of a predomain $A$. Then $\textit{Adm}(P)$ is
the least admissible superset of $P$. Concretely, $a\in
\textit{Adm}(P)$ iff there exists a chain $(a_i)_i$ such that $a_i\in
P$ for all $i$ and $a=\sup_i a_i$.
\end{definition}
\begin{lemma}\label{funad}
If $f:A_1\times\dots\times A_n$ is continuous; $P_i\subseteq A_i$ are
arbitrary subsets and $Q\subseteq B$ is admissible then
$f:P_1\times\dots \times P_n\rightarrow Q$ implies
$f:\textit{Adm}(P_1)\times\dots \times \textit{Adm}(P_n)\rightarrow
Q$.
\end{lemma}
\begin{lemma}\label{func}
Let $A,B$ be predomains and let $(p_i)_i$ be a chain of retracts on
$B$ such that $p_i(b)$ is compact for each $i$ and $\sup_i p_i =
\idd_B$ and $b\in Q$ implies $p_i(b)\in Q$ for all $i$.
Then $P{\rightarrow}\textit{Adm}(Q)=\textit{Adm}(P\rightarrow Q)$.
\end{lemma}

\subsection{Traces}
A trace models a terminating run of a concurrent computation as
a sequence of pairs of heaps, each representing
pre- and post-state of one or more atomic actions. The semantics of a
program then is a (typically large) set of traces (and
final values), accounting
for all possible environment interactions. 

\begin{definition}[Traces]
A trace is a finite sequence of the form $(\h_1,\k_1) (\h_2,\k_2) \cdots (\h_n,\k_n) $ where for $1 \leq j \leq i \leq n$, we have $\h_i,\k_i \in\Stores$  and $\dom{\h_j} \subseteq \dom{\h_i}, \dom{\h_j} \subseteq \dom{\k_i}, \dom{\k_j} \subseteq \dom{\h_i}, \dom{\k_j} \subseteq \dom{\k_i}$. We write $\textit{Tr}$ for the set of traces. 
\end{definition}

Let $ t$ be a trace. A trace of the form $u\,(\h,\h)\,v$ where $ t= u v$ is said to
arise from $t$ by stuttering. A trace of the form $u(\h,\k)v$ where
$t=u(\h,\heapp)(\heapp,\k)v$ is said to arise from t by mumbling. For example, if
$t=(\h_1,\k_1)(\h_2,\k_2)(\h_3,\k_3)$ then $(\h_1,\k_1)(\h,\h)(\h_2,\k_2)(\h_3,\k_3)$ 
arises from $t$
by stuttering. In the case where $\k_1=\h_2$ the trace $(\h_1,\k_2)(\h_3,\k_3)$
arises from $t$ by mumbling. A set of traces $U$ is closed under stuttering and
mumbling if whenever $t'$ arises from $t$ by stuttering or mumbling and
$t\in U$ then $t'\in U$, too. 

Brookes~\cite{brookes96ic} gives a fully-abstract semantics for while-programs
with parallel composition using sets of traces closed under stuttering
and mumbling. We here extend his semantics to higher-order functions
and general recursion.

\begin{definition}[Trace Monad]
\label{def:monad}
Let $A$ be a predomain. Elements of the domain $TA$ 
  are sets $U$ of pairs $\tup{t,a}$ where $t$ is a
 trace and $a\in A$ such that the following properties are satisfied:
\begin{itemize}
\item \textit{[S\&M]}: if $t'$ arises from t by stuttering or mumbling and
 $\tup{t,a}\in U$ then $\tup{t',a} \in U$. 
\item \textit{[Down]}: if $\tup{t,a_1}\in U$ and $a_2\leq a_1$ then
  $\tup{t,a_2}\in U$.
\item \textit{[Sup]}: if $(a_i)_i$ is a chain in $A$ and $\tup{t,a_i}\in U$ for all $i$ then  $\tup{t,\sup_i a_i}\in U$. 
\end{itemize}
The elements of $TA$ are partially ordered by inclusion.  
\end{definition}
\begin{lemma}
If $A$ is a predomain then $TA$ is a domain. 
\end{lemma}
An element $U$ of $TA$ represents the possible outcomes of a
nondeterministic, interactive computation with final result in
$A$. Thus, if $\tup{t,a}\in U$ for $t=(\h_1,\k_1)\dots(\h_n,\k_n)$
then there could be $n$ interactions with the
environment with heaps $\h_1,\dots,\h_n$ being ``played'' by the
environment and ``answered'' with heaps $\k_1,\dots,\k_n$ by the
computation. After that, this particular computation ends and $a$ is
the final result value.

For example, the semantics of 
$\cloc:=\myread{\cloc}+1;\cloc:=\myread{\cloc}+1;\myread{X}$ 
contains many traces, including the following, where we write $[n]$ for
the heap in which $\cloc$ has value $n$:

\begin{tabular}{l}
$\tup{([10],[12]),12}$, \\
$\tup{([10],[11])([15],[16]),16}$,\\
$\tup{([10],[11])([15],[16])([17,17]),17}$,\\
$\tup{([10],[11])([15],[16])([17,17]),16}$,\\
$\tup{([10],[11])([17],[17])([15],[16]),16},\ldots$
\end{tabular}

Axiom [S\&M] is taken from Brookes. It ensures that the semantics does
not distinguish between late and early choice \cite{dreyer} and
related phenomena which are reflected, e.g., in resumption semantics
\cite{plotkin76siam}, but do not affect observational equivalence.
Note that non-termination is modelled by the empty set, so we are
working with an `angelic' notion of equivalence
(`may semantics'
\cite{DBLP:dblp_conf/icalp/NicolaH83}). For example, the semantics of
 $\assign{X}{0};\myif{X\mequals 0}{0}{\keywd{diverge}}$
is the same as that of $\assign{X}{0};0$ and contains, for example
$\tup{([10],[0]),0}$ but also (stuttering)
$\tup{(([10],[0]),([34],[34])),0}$.  Note that it is not possible to
tell from a trace whether an external update of $\cloc$ has happened
before or after the reading of $\cloc$.

Let us also illustrate how traces iron out some intensional
differences that show up when concurrency is modelled using transition
systems or resumptions. Consider the following two programs where
$\mnond{}$ denotes a nondeterministically chosen boolean value.
\[
\begin{array}{lll}
e_1 &\equiv& \mif{~\mnond{}}\ \mthen\ \assign{X}{0};\mtrue\ \melse\ 
\assign{X}{0};\mfalse\\ 
e_2 &\equiv&\assign{X}{0};~\mnond{}
\end{array}
\]
Both $e_1$ and $e_2$ admit the same traces, namely
$(([x],[0]),\mtrue)$ and $(([x],[0]),\mfalse)$ and stuttering variants
thereof.  In semantic models based on transition systems or
resumptions and bisimulation, these are distinguished, which
necessitates the use of special mechanisms such as history and
prophecy variables \cite{DBLP:journals/tcs/AbadiL91}, forward-backward
simulation \cite{DBLP:dblp_journals/iandc/LynchV96}, or speculation
\cite{dreyer} in reasoning.

Axioms [Down] and [Sup] are known from the Hoare powerdomain
\cite{plotkin76siam}. Recall that the Hoare
powerdomain $PA$ contains the subsets of $A$ which are downclosed
([Down]) and closed under suprema of chains ([Sup]). Such subsets are
also known as Scott-closed sets.  Thus, $TA$ is the restriction of
$P(\textit{Tr}\times A)$ to the sets closed under stuttering and
mumbling. Axiom [Down] ensures that the ordering is indeed a partial
order and not merely a preorder. 
Additional nondeterministic outcomes that are less
defined than existing ones are not recorded in the semantics.

\begin{definition}
If $U\subseteq\textit{Tr}\times A$ then $U^\dagger$ is the
least subset of $TA$ containing $U$, i.e.\ $U^\dagger$ is the closure
of $U$ under [S\& M], [Down], [Sup].
\end{definition}

\begin{definition}
Let $A,B$ be a predomains.  We define the continuous
functions $\textit{rtn} : A \rightarrow TA$ and $\textit{bnd} :
(A{\rightarrow} TB)\times TA\rightarrow TB$ by:
\[\begin{array}{l}
 \textit{rtn}(a):= (\{\tup{(\h,\h),a} \mid \h\in\Stores\})^\dagger\\ 
\textit{bnd}(f,g) := (\{\tup{uv,b}\mid \tup{u,a}\in g\wedge \tup{v,b}\in f(a)\})^\dagger\\
\end{array}\]
\end{definition}
\noindent
These endow $TA$ with the structure of a strong monad. The continuous function
%
$\textit{fromstate} : SA\rightarrow TA$ is defined by:
\[
\begin{array}{l}
 \textit{fromstate}(c) := \{\tup{(\h,\k),a}\mid 
c(\h)=\tup{\k,a}\}^\dagger
\end{array}
\]
If $t_1,t_2,t_3$ are traces, we write $\textit{inter}(t_1,t_2,t_3)$ to
mean that $t_3$ can be obtained by interleaving $t_1$ and $t_2$ in
some way, i.e., $t_3$ is contained in the shuffle of $t_1$ and $t_2$.
In order to model parallel composition we introduce the following
helper function
\[
\begin{array}{l}
 \semparallel :~ TA\times TB\rightarrow T(A\times B)\\
U\semparallel V := \{\tup{t_3,(a,b)}\mid \textit{inter} (t_1,t_2,t_3), \tup{t_1,a}\in U, \tup{t_2,b}\in V\}^\dagger
\end{array}
\]
The continuous map $\textit{at} : TA\rightarrow TA$ is defined by:
\[
\begin{array}{l}
 \textit{at}(U) := \{\tup{(\h,\k),v}\mid ((\h,\k),v)\in U\}^\dagger
\end{array}
\]
\noindent
Notice that due to mumbling $\tup{(\h,\k),v}\in U$ iff
there exists an element $\tup{(\h_1,\h_2)(\h_2,\h_3)\dots
  (\h_{n-2},\h_{n-1})(\h_{n-1},\h_n),v}\in U$ where $\h=\h_1$ and
$\h_n=\k$. The presence of such an element, however, models an atomic
execution of the computation represented by $U$.

\subsection{Semantic values}
The predomain $\Values$ of untyped values
is the least solution of the following domain equation: 
\[
\Values \simeq \ValuesB + (\Values\rightarrow T\Values) + \Values^*.
\]
That is, values are either R-values, continuous functions from values to computations ($T\Values$), or tuples of values.
We tend to identify the summands of the right hand side with subsets
of $\Values$ but may use tags like $\funn f\in\Values$ when
$f:\Values\rightarrow T\Values$ to avoid ambiguity.

We have families of deflations $p_i:\Values\partfun\Values$ and $q_i:
T\Values\rightarrow T\Values$, referred to as canonical deflations, so
that $(p_i)_i$ and $(q_i)_i$ are ascending chains converging to the
identity. The definition is entirely standard and may be found in the
accompanying material. It shows in particular that $\Values$ and
$T\Values$ are \emph{bifinite} (equivalently SFP) (pre-)domains
\cite{Abramsky94domaintheory} and as such also Scott (pre-)
domains. The presence of these deflations allows us to apply
Lemma~\ref{func} and simplifies reasoning in general.

The semantics of values $\semV{v}\in\Values\rightarrow\Values$ and
terms $\semC{t}\in\Values\rightarrow T\Values$ are given by the
recursive clauses in Figure~\ref{seme}. Environments, $\rho$, are properly tuples of values; we abuse notation slightly by treating them as maps from variables, $x$, to values, $v$, (and write $\rho[x{\mapsto}v]$ for functional update) to avoid mentioning an explicit context in which untyped terms are well-formed.


\begin{figure*}[tph]
\vspace{-10mm}
\[
\begin{array}{@{}rcl}
 \semV{x}\rho &=& \rho(x)\\
\semV{v_r}\rho &=& v_r\\
\semV{(v_1,v_2)}\rho &=& (\semV{v_1}\rho,\semV{v_2}\rho)\\
\semV{v.i}\rho &=& d_i\ \mbox{if $i=1,2$, $\semV{v}\rho = (d_1,d_2)$}\\
\semV{c}\rho &=& \funn f\\
&\multicolumn{2}{l}{\mbox{ where
  $f(v)=\textit{rtn}(F_c(v))$ if $F_c(v)$ is defined}}
  \\ & \multicolumn{2}{l}{\textrm{  and $f(v)=\emptyset$, otherwise.}}
\\
\semV{\vfix{f}{x}e}\rho &=& \funn
{g^\ddagger(\rho)}\\&\multicolumn{2}{l}{\mbox{ where
  $g(\rho,u)=\lambda d.\semV{e}\rho[f{\mapsto}u,x{\mapsto}d]$}}
\\
\semV{v}\rho &=& 0\mbox{, otherwise}
\end{array}
\begin{array}{rcl}
 \semC{v}\rho&=&\textit{rtn}(\semV{v}\rho)\\
\semC{\letin{x}{e_1}{e_2}}\rho &=& \textit{bnd}(\lambda d.\semC{e_2}\rho[x{\mapsto}d],\semC{e_1}\rho)\\
\semC{v_1\ v_2}\rho &=& \semV{v_1}\rho(\semV{v_2}\rho)\\
\semC{\myif{v}{e_1}{e_2}}\rho &=& \semC{e_1}\rho\mbox{, if }\semV{v}\rho=\mtrue\\  
\semC{\myif{v}{e_1}{e_2}}\rho &=& \semC{e_2}\rho\mbox{, if }\semV{v}\rho=\mfalse\\ 
\semC{\myread{v}}\rho\  &=& \textit{fromstate}(\lambda \heap.(\heap,\heap(\cloc)))\mbox{, when
$\semV{v}\rho=\cloc$}\\
\semC{\assign{v_1}{v_2}}\rho\  &=&\textit{fromstate}(\lambda \heap.
(\heap[\cloc{\mapsto}\semV{v_2}\rho],()))\mbox{, if $\semV{v_1}\rho=
\cloc$}\\
\semC{\myref{v}}\rho\  &=& \textit{fromstate}(\lambda \heap.\textit{new}(\heap,\semV{v}\rho))\\
\semC{\myatomic{e}}\rho &=& \textit{at}(\semC{e})\\
\semC{\mypar{e_1}{e_2}}\rho &=& \semC{e_1}\rho\semparallel \semC{e_2}\rho\\
\semC{e}\rho\ &=& \emptyset \mbox{, otherwise}
\end{array}
\]
\caption{Denotational semantics \label{seme}}
\end{figure*}
\section{Abstract Locations}
We build on the concept of abstract locations defined by Benton,
Hofmann, and Nigam \cite{DBLP:dblp_conf/popl/Benton0N14}. These allow
complicated data structures that span several concrete locations,
or only parts of them, to be a regarded as a single ``location'' that
can be written to and read from. Essentially, an abstract location is
given by a partial equivalence relation on heaps modelling well-formedness and 
 equality together with a transitive relation modelling allowed
modifications of the abstract location. Abstract locations then allow
certain commands that modify the physical heap to be treated as 
read-only or even pure if they respect the contracts. Abstract locations are
related to \emph{islands} \cite{DBLP:conf/popl/AhmedDR09} which also
  allow one to specify heap allocated data structures and use
  transition systems for that purpose. An important difference is that
  abstract locations do not require physical footprints in the form of
  sets of concrete locations.

Due to the absence of dynamic allocation at the level of abstract
locations in the present paper, we can slightly simplify the original definition 
\cite{DBLP:dblp_conf/popl/Benton0N14}, dropping those axioms that involve the interaction with dynamic allocation.\footnote{Though our examples do all satisify these axioms, leaving the way open to a future extension with dynamically allocation of abstract locations and concurrency.}
On the other hand, in the presence of concurrency, we need \emph{two}
partial equivalence relations: one that  models semantic
equivalence and well-formedness and a finer one that constrains the
heap modifications that other concurrent computations that are
independent of the given abstract locations are allowed to do
\emph{while} an operation on the abstract location is ongoing, but
temporarily preempted.

\begin{definition}[Concurrent Abstract Location]\label{absloc}
  A \emph{concurrent abstract location} $\loc$ consists of the following data:

(1) a partial equivalence relation  $\rloc{}{}{\loc}$ on
  $\Stores$ modeling the ``semantic equivalence'' on the bits of the
  store that $\loc$ uses.
If $\heap\rloc{}{}{\loc}\heap'$ then the same computation started on $\heap$ and $\heap'$, respectively, will yield related or even equal results.
 
(2) a partial equivalence relation  $\rrloc{}{}{\loc}$ on
  $\Stores$ refining $\rloc{}{}{\loc}$ and modeling the ``strict equivalence'' on the bits of the
  store that $\loc$ uses. 
If a concurrent computation on $\loc$ has reached $\heap$ and is preempted, then another computation may replace $\heap$ with $\heap'$ where $\heap\rrloc{}{}{\loc}\heap'$ and then the original computation on $\loc$ may resume on $\heap'$ without the final result being compromised. 

(3) a transitive  (and reflexive on the support of $\rloc{}{}{\loc}$) 
 relation $\gloc{}{}{\loc}$ modeling how exactly the
  heap may change upon writing the abstract location and in particular
  what bits of the store such writes leave intact. In other words, if
  $\gloc{\heap}{\heap_1}{\loc}$ then $\heap_1$ might arise by writing
  to $\loc$ in $\heap$ and all possible writes are specified by
  $\gloc{}{}{\loc}$. We call $\gloc{}{}{\loc}$ the \emph{step relation} of $\loc$. 

In addition, we require the following 
conditions where $\inR{\heap}{\loc}$ stands for $\rloc{\heap}{\heap}{\loc}$.
\begin{enumerate}
  \item If $\inR{\heap}{\loc}$ then $\rrloc{\heap}{\heap}{\loc}$;
  \item if $\gloc{\heap}{\heap_1}{\loc}$ then $\inR{\heap}{\loc}$ and $\inR{\heap_1}{\loc}$.
\end{enumerate}
If $\gloc{\h}{\h_1}{\loc}$ and at the same time
$\rrloc{\h}{\h_1}{\loc}$, then we say that $\h_1$ arises from $\h$ by a \emph{silent move} in $\loc$. Our semantic framework will permit silent
moves at all times.
\end{definition}
We now introduce some examples of abstract locations. 

\label{sec:abs-examples}
\paragraph{Single Integer} For our simplest example, consider the following abstract location parametric with respect to concrete location $X$ as follows:
\[
 \begin{array}{lcl}
 \rloc{\h}{\h'}{\locInt(\cloc)}
  &\iff& \exists n.\heap(X) = \intt n \land \heap'(X) = \intt n\\
\rrloc{\h}{\h'}{\locInt(\cloc)} &\iff&  \rloc{\h}{\h'}{\locInt(\cloc)}\\
  
  \gloc{\h}{\h_1}{\locInt (\cloc)} &\iff& \\ 
  \multicolumn{3}{r}{\qquad \inR{\heap}{\locInt(\cloc)}, \inR{\heap_1}{\locInt(\cloc)} \textrm{ and } \forall \cloc' \in \labs. \cloc' \neq X \Rightarrow \heap(\cloc') = \heap_1(\cloc)}

 \end{array}
\]
Two heaps are semantically equivalent (w.r.t.\ $\locInt(\cloc)$ that
is) if the values stored in $X$ are integers and equal; the step
relation requires all other concrete locations to be unchanged.

We will sometimes abuse notation and write $\rEff{\cloc}, \wEff{\cloc}, \cEff{\cloc}$ for $\rEff{\locInt(\cloc)}, \wEff{\locInt(\cloc)}, \cEff{\locInt(\cloc)}$. 

\paragraph{Overlapping references}
Let $\cloc$ be a concrete location encoding a pair of integer values using a bijection $p$.
We define the abstract location $\locLo(\cloc)$ as below. We omit $\locHi(\cloc)$ which is similar, but only looks at the second projection, instead of the first. 
\[
 \begin{array}{l}
  \rloc{\h}{\h'}{\locLo(\cloc)}  \iff 
   \exists \vval_1\vval_2\vval_1'\vval_2'\in\ints.
\heap(X) = p^{-1}(\vval_1, \vval_2) \land ~ \\ \qquad \qquad \qquad \quad \heap'(X) = p^{-1}(\vval_1', \vval_2') \land \vval_1 = \vval_1'\\[2pt]


\rrloc{\h}{\h'}{\locLo(\cloc)} \iff  \rloc{\h}{\h'}{\locLo(\cloc)} \\[2pt] 
  \gloc{\h}{\h_1}{\locLo(\cloc)}  \iff \inR{\heap}{\locLo(\cloc)}, \inR{\heap_1}{\locLo(\cloc)} \textrm{ and }\\
  \quad (\forall \cloc' \neq \cloc.\heap(\cloc')=\heap_1(\cloc'))\wedge 
(\forall\vval_1\vval_2\vval_1'\vval_2'\in\ints.
\heap(X) = p^{-1}(\vval_1,\vval_2) \land~ \\ \qquad \qquad  \heap_1(X) = p^{-1}(\vval_1', \vval_2') \Rightarrow\vval_2 = \vval_2')\\[2pt]
  
 \end{array}
\]
The semantic (and strict) equivalence of $\locLo(\cloc)$ (respectively,
$\locHi(\cloc))$ specifies that two heaps $\heap$ and $\heap'$ are
equivalent whenever they both store a pair of values in $\cloc$ and
the first projections (respectively, second projection) of these pairs
are the same.  The step relation of $\locLo(\cloc)$ (respectively,
$\locHi(\cloc)$) specifies that it keeps all other locations alone and
does not change the second projection (respectively, first projection)
of the pair stored at location $\cloc$.

\paragraph{Version Numbers} 
The abstract location $\mathfrak{X}$ consists of two concrete locations $X_{Val}$ and $X_{Ver}$ and its relations are specified as follows:
\[
\hspace{-4mm}
 \begin{array}{ll}
 \rloc{\h}{\h'}{\mathfrak{X}}
  \iff &   \h(\cloc_{Val}) = \h'(\cloc_{Val})\\
\rrloc{\h}{\h'}{\mathfrak{X}} \iff&  \rloc{\h}{\h'}{\mathfrak{X}}\\
    \gloc{\h}{\h_1}{\mathfrak{X}} \iff& \forall \cloc' \notin \{\cloc_{Ver},\cloc_{Val}\}. \h(\cloc') = \h_1(\cloc') ~ \land~\\
    & \inR{\h}{\mathfrak{X}} \land \inR{\h_1}{\mathfrak{X}} \land \h(X_{Ver}) <= \h_1(X_{Ver})~  \land~ \\
    &[\h(X_{Val}) \neq \h_1(X_{Val}) \Rightarrow \h(X_{Ver}) < \h_1(X_{Ver})]
 \end{array}
\]
Two heaps are semantically equivalent if they have the same value (independent of the version number). The step relation specifies that the version number does not descrease and it increases if the value changes.

\paragraph{Loop Parallelization} For a concrete location $\cloc$, we introduce two concurrent abstract locations $\locListEven(\cloc)$ and $\locListOdd(\cloc)$, which only look, respectively, at the elements in the the even and odd positions of the linked list pointed by $\cloc$. Formally, let $L(\cloc,\h)$ denote that $\h(\cloc)$ points to a well formed linked list of integers of length $L(\cloc,\h).len$ and locations $L(\cloc,\h).locs$ and that $L(\cloc,\h)[i]$ is the $i^{th}$ node of the list for $1 \leq i \leq L(\cloc,\h).len$. The relations for $\locListEven(\cloc)$ are as below. We omit the relations for $\locListOdd(\cloc)$, which are similar.
\[
\hspace{-4mm}
 \begin{array}{ll}
 \rloc{\h}{\h'}{\locListEven(\cloc)}
  \iff &   L(\cloc,\h) \land L(\cloc,\h') \land L(\cloc,\h).len = L(\cloc,\h').len ~\land\ \\ 
  & L(\cloc,\h)[2i] = L(\cloc,\h')[2i]\\ 
  &  \textrm{for } 0 \leq i \leq \lfloor L(\cloc,\h).len / 2 \rfloor
  \\
\rrloc{\h}{\h'}{\locListEven(\cloc)} \iff&  \rloc{\h}{\h'}{\locListEven(\cloc)}\\
    \gloc{\h}{\h_1}{\locListEven (\cloc)} \iff& 
  \inR{\h}{\locListEven (\cloc)} \land \inR{\h_1}{\locListEven (\cloc)}~ \land \\
  & L(\cloc,\h) \land L(\cloc,\h_1) \land \textrm{for } 0 \leq i \leq \lfloor L(\cloc,\h).len / 2 \rfloor\  \\
  & \quad L(\cloc,\h)[2i+1] = L(\cloc,\h_1)[2i+1]~ \land \\ 
  & \quad L(\cloc,\h)[2i].next = L(\cloc,\h_1)[2i].next~\land \\
  & \forall \cloc' \notin L(\cloc,\h).locs. \h(\cloc') = \h_1(\cloc)
 \end{array}
\]

\paragraph{Michael-Scott queue}
For concrete location $\cloc$ we introduce a concurrent abstract location $\locMSQ(\cloc)$ first informally as follows: we have $\rloc{\heap}{\heap'}{\locMSQ(\cloc)}$ if both $\heap$ and $\heap'$ contain a well-formed MSQ rooted at $\cloc$ and these queues contain the same entries in the same order. They may, however, use different locations for the nodes and also have different garbage tails. 

The relation $\rrloc{\heap}{\heap'}{\locMSQ(\cloc)}$ asserts that $\heap$ and $\heap'$ are identical on the part reachable and co-reachable from $\cloc$ via $\textit{next}$ pointers. This means that while an MSQ operation is working on the queue no concurrent operation working elsewhere is allowed to relocate the queue or remove the garbage trail which would be the case if we merely required that such operations do not change the 
$\rloc{}{}{MSQ(\cloc)}$-class.

The relation $\gloc{}{}{\locMSQ(\cloc)}$, finally, is defined as the
transitive closure of the actions of operations on the MSQ: adding
nodes at the tail and moving nodes from the head to the garbage tail.

We now give a formal definition. We represent pointers \emph{head},
\emph{next}, \emph{elem} using some layout convention,
e.g.\ $v.\emph{head}=v.1$, etc. We then define
 \[
 \begin{array}{ll}
  \h,\cloc\stackrel{\mathit{next}}{\to} \cloc' \iff & \cloc' \textrm{ can be reached 
  from $\cloc$ in $\h$} \\ &
  \textrm{by following a chain of next pointers}  
 \end{array}
 \]
We use
$\textit{List}(\cloc,\h,(\cloc_0,\ldots,\cloc_n),(v_1\ldots,v_n))$ to
signal that $\h(\cloc)$ points to a linked list with nodes
$X_0,\dots,X_n$ and entries $v_1,\dots v_n$. Note that the first node
$X_0$ acts as a sentinel and its $elem$ component is
ignored. Formally:
 \[
 \hspace{-3mm}
  \begin{array}{l@{\quad}l}
   \h(\cloc).head = \cloc_0 &
   \h(\cloc_i).elem = v_i \textrm{ for $i=1,\ldots, n$}\\
   \h(\cloc_i).next = \cloc_{i+1}  \textrm{ for $i=0,\ldots, n-1$}&
   \h(\cloc_n).next = null
  \end{array}
 \]
We define $\textit{fp}(\cloc,\h)$ as the set of locations reachable and co-reachable from $\cloc$ via $\textit{next}$, formally: 
\[
\textit{fp}(\cloc,\h)=\{\cloc'\mid \cloc\stackrel{\mathit{next}}{\to} \cloc'\vee\cloc'\stackrel{\mathit{next}}{\to} \cloc\}
\]
Finally, we define $\textit{snoc}(\h,\h',\cloc,v)$ to mean that $\h'$ arises from $\h$ by attaching a new node containing $v$ at the end of the list pointed to by $\cloc$ in $\h$. Thus, in particular, $\textit{List}(\cloc,\h,(\cloc_0,\ldots,\cloc_n),(v_1\ldots,v_n))$ implies $\textit{List}(\cloc,\h',(\cloc_0,\ldots,\cloc_n,\cloc_{n+1}),(v_1\ldots,v_n,v))$ for some $\cloc_{n+1}\not\in\dom{\h}$. We omit the obvious frame conditions. 
We now define
\[
 \begin{array}{lcl}
  \rloc{\h}{\h'}{\locMSQ(\cloc)} &\iff&  \exists \vec{X}\ \vec{X'}\ \exists \vec{v}.\textit{List}(\cloc,\h,\vec{X},\vec{v}) \land \textit{List}(\cloc,\h',\vec{X'},\vec{v})\\[1pt]
\rrloc{\h}{\h'}{\locMSQ(\cloc)} &\iff&  \rloc{\h}{\h'}{\locMSQ(\cloc)} \wedge
\forall\cloc'\in\textit{fp}(\cloc,\h).\h(\cloc')=\h'(\cloc')\\[1pt]
\gloc{\h}{\h_1}{\locMSQ(\cloc)} &\iff& \inR{\h}{\locMSQ(\cloc)}\wedge
\inR{\h_1}{\locMSQ(\cloc)}\wedge \textit{step}^*(\h,\h_1)\\[1pt]
\textit{step}(\h,\h_1)&\iff& \forall \cloc'\neq\cloc.\h(\cloc')=\h_1(\cloc')~ \land \\ && [\h_1(\cloc)=\h(\cloc).\textit{next}\vee \exists v.\textit{snoc}(\h,\h_1,\cloc,v)]
\end{array}
\]

In all of these examples, the only silent moves are identity moves.  This is not so in the examples from \cite{DBLP:dblp_conf/popl/Benton0N14} which contained  data-structures that would reorganize during lookups and also patterns like late initialisation. 
\subsection{Worlds}
We will group the abstract locations used to describe a program into a
\emph{world}. In this paper we do not model dynamic evolution of
worlds; all abstract locations ever used must be set up upfront. While
allocation of concrete locations may happen to increase a data structure modelled by
an abstract location, e.g.\ in the Michael-Scott Queue example, no new
such datastructures can appear. It is possible, however, to extend our
work in this direction by using (proof-relevant) Kripke logical
relations
\cite{DBLP:dblp_conf/popl/Benton0N14,DBLP:conf/popl/AhmedDR09}.
\begin{definition}[world]
A \emph{world} is a set of abstract locations. 

The relation $\heap\models \w$ (heap $\h$ satisfies world $\w$) is defined as the largest relation such that $\h\models \w$ implies 
\begin{itemize}
\item \inR{\h}{\loc} for all $\loc\in\w$; 
\item if $\loc\in\w$ and $\gloc{\h}{\h_1}{\loc}$ then 
$\rrloc{\heap}{\heap_1}{\loc'}$ holds for all 
$\loc'\in\w$ with $\loc'\neq\loc$ and $\h_1\models\w$. 
\end{itemize}
\end{definition}
The original account of abstract locations
\cite{DBLP:dblp_conf/popl/Benton0N14} also has a notion of
independence of locations which facilitates reasoning in the presence
of dynamic allocation, and in particular permitted relocation of
abstract locations. Since we are not currently treating dynamic
allocation of abstract locations, we can avoid this notion here.

We remark that if our world $\w$ contains two obviously ``dependent'' abstract locations, e.g.\ has  both an integer location and a boolean location
placed at the same physical location, then there will be no heap $\h$ such that $\h\models\w$. 

We assume a fixed \emph{current} world $\w$ which may appear in definitions without being notationally reflected. See also Assumption~\ref{assi}.
\section{Effects}
For each abstract location $\loc$ we have three elementary effects
$\rEff\loc$ (reading from $\loc$), $\wEff\loc$ (writing to $\loc$),
and $\cEff\loc$ (chaotic or concurrent access). The chaotic access is
similar to writing, but allows writes that are not in sync. For
example, $e_1 = \assign{X}{1}$ and $e_2 = \assign{X}{2}$ both have
individually the $\wEff{X}$ effect, but $e_1$ and $e_2$ are
distinguishable with a context that assumes the
$\wEff{X}$-effect. Thus, $e_1$ and $e_2$ are not equal ``at type''
$\wEff{X}$. At type $\cEff{X}$ they are, however, equal, because a
context that copes with this effect may not assume that both produce
equal results.

We use the $\cEff{\loc}$ effect to tell the environment not to look at
a particular location during a concurrent computation. For example, we
will be able to show that $\assign{X}{\myread{X}+1};\assign{X}{
  \myread{X}+1}$ is equivalent to $\assign{X}{ \myread{X}+2}$ ``at
type'' $\myety{\unittype}{\cEff{X}}{\eff}{\eff\cup\{\rEff{X},\wEff{X}\}}$ whenever $X\notin\locs{\eff}$. This means that
the two computations are indistinguishable by environments that do not
read, let alone modify $X$ during the computation and assume regular
read-write access once it is completed.
It would alternatively be possible to replace the $\cEff{}$-effect
using a special set of private locations akin to the private regions
from \cite{birkedal}.

We use the notation $\rds(\eff)$, $\wrs(\eff)$, $\concs{\eff}$ to
refer to the abstract locations $\loc$ for which $\eff$ contains
$\rEff\loc$, $\wEff\loc$, and $\cEff\loc$, respectively. We write
$\locs\eff:=\rds(\eff)\cup\wrs(\eff)\cup\concs{\eff}$.  We also write
$\eff^C$ for $\eff$ with all read effects removed and each $\wEff\loc$ in $\eff$ replaced by
$\cEff\loc$.

\begin{definition}
An effect $\eff$ is well-formed (with respect to the current world) if  $\locs\eff\subseteq\w$ and $\rdsin\eff\cap\concs\eff=\emptyset$ and $\concs\eff\subseteq\wrsin\eff$. 
An effect specification is a triple $(\eff_1,\eff_2,\eff_3)$ of well-formed effects such that $\eff_2\subseteq\eff_3$. 
\end{definition}
An effect specification $(\eff_1,\eff_2,\eff_3)$ approximates the
behaviour of a computation $e$ in the following way: the effect
$\eff_1$ summarizes side effects that may occur during the execution
of $e$ (corresponding to a guarantee condition in the rely-guarantee
formalism \cite{DBLP:journals/logcom/ColemanJ07}); the effect $\eff_2$ summarizes effects of the interacting
environment that $e$ can tolerate while still functioning as expected
(corresponding to a rely condition). Finally, $\eff_3$
summarizes the side effects that may occur between start and completion of $e$. All the effects that the
environment might introduce must be recorded in $\eff_3$ because they are
not under ``our'' control and might happen at any time even as the
very last thing before the final result is returned. The effects
flagged in $\eff_1$, on the other hand, do not necessarily show up in
$\eff_3$, for a computation might be able to clean up those effects
prior to returning the final result. The requirement that $\rdsin{\eff}\cap \concs{\eff}=\emptyset$ is owed to the fact that all effects should preserve their own precondition, however the precondition of $\rEff{\loc}$ is agreement on $\loc$ which is not preserved by $\cEff{\loc}$. The requirement $\concs{\eff}\subseteq\wrs{\eff}$ reflects the fact that $\concs{\loc}$ includes $\wEff{\loc}$ as a special case. 

Note that if $\eff^C \cup \eff_1$ is a (well-formed) effect, then it
is the case that $\rdsin{\eff_1} \cap (\wrsin{\eff} \cup \concs{\eff})
= \emptyset$. We will use this observation to simplify some side
conditions.

In our concrete examples, we abbreviate $\{\cEff{\loc}\}\cup\{\wEff{\loc}\}$ by just $\cEff{\loc}$, in other words, the chaotic effect silently implies the write effect.

Consider the computations $e_1 =
\assign{X}{\myread{X}+1};\assign{X}{\myread{X}+1}$ and $e_2 =
\assign{X}{\myread{X}+2}$.  Let $\eff_X$ stand for
$\{\rEff{X},\wEff{X}\}$ and analogously $\eff_Y$.  Each of the two
computations can be assigned the effect $(\eff_X,\eff_Y,\eff_X\cup
\eff_Y)$, but they are distinguishable at that effect typing. Under
the looser specification $(\{\cEff{\eff_X}\},\eff_Y,\eff_X\cup
\eff_Y)$, however, they are indistinguishable, and our semantics is
able to validate this equivalence, see Example~\ref{jife}.

Finally, consider the program $e = \myread{\cloc}$ that simply reads a location storing an integer. We can show that this program has type $\myety{\ints}{\emptyset}{\eff}{\eff, \rEff{X}}$, where the read effect on $\cloc$ is only in the global effects. 
 
\paragraph{Notations.}
For any well-formed effects $\eff, \eff'$ we use the notation
$\eff\perp \eff'$ to mean that $\rdsin{\eff} \cap \wrsin{\eff'} =
\rdsin{\eff'} \cap \wrsin{\eff} = \wrsin{\eff} \cap \wrsin{\eff'} =
\emptyset$. Note that this implies in particular
$\concs{\eff}\cap\rdsin{\eff'}=\emptyset$, etc. Intuitively, two
programs exhibiting effects $\eff$ and $\eff'$, respectively, commute
with each other.  We write $\rloc{\h}{\h'}{\rdsin{\eff}}$ to mean
$\rloc{\h}{\h'}{\loc}$ for each $\loc\in \rdsin{\eff}$.  We write
$\gloc{}{}{\eff}$ for the transitive closure of
$\bigcup_{\loc\in\wrsin{\eff}}\gloc{}{}{\loc}\cup\bigcup_{\loc\in\w}\gloc{}{}{\loc}\cap\rrloc{}{}{\loc}$. Thus,
$\gloc{}{}{\eff}$ allows steps by locations recorded as writing in
$\eff$ and silent steps by all locations in the current world.

We define the notation $\eff_1\sqcup\eff_2$ which appears in the parallel congruence rule by
\[
\eff_1\sqcup\eff_2 = \eff_1\cup\eff_2\setminus\{\wEff\ell\mid 
\wEff\ell\not\in\eff_1\cap\eff_2\}\setminus\{\cEff\ell\mid 
\cEff\ell\not\in\eff_1\cap\eff_2\}
\]
\section{Typing and congruence rules}

Types are given by the grammar 
\[
\tau ::= \unittype \mid \inttype\mid \booltype \mid A
\mid\tau_1\times\tau_2\mid 
\tau_1\myeffto{\eff_1}{\eff_2}{\eff_3} \tau_2
\] 
where $A$ ranges over user-specified abstract types. They will typically 
 include reference types such as $\keywd{intref}$
and also types like lists, sets, and even
objects. In $\tau_1\myeffto{\eff_1}{\eff_2}{\eff_3} \tau_2$ the triple of 
effects $(\eff_1,\eff_2,\eff_3)$ must be an effect specification. 

We use two judgments:
\begin{itemize}
\item $\Gamma \vdash v\leq v' : \tau$ specifying that values $v$ and $v'$ have
 type $\tau$ and that $v$ approximates $v'$,
 \item $\Gamma \vdash e\leq e' : \myety{\tau}{\eff_1}{\eff_2}{\eff_3}$ specifying that the programs $e$ and $e'$ 
under the context $\Gamma$ have type $\tau$, with the effect specification $(\eff_1, \eff_2, \eff_3)$ specifying, respectively, the effects during execution, the effects of the interacting environment and the start and completion effects. Moreover, $e$ approximates $e'$ at this specification. 
\end{itemize}

We assume an ambient set of \emph{axioms} each having the form
$(v,v',\tau)$ where $v, v'$ are values in the metalanguage and $\tau$
is a type meaning that $v$ and $v'$ are claimed to be of type $\tau$
and that $v$ approximates $v'$. This must then be proved ``manually''
using the semantics rather than using the rules. The

We also define typing judgements $\Gamma\vdash v:\tau$ and
$\Gamma\vdash e:\myety{\tau}{\eff_1}{\eff_2}{\eff_3}$ which denote the
special case when $\Gamma\vdash v\leq v:\tau$ and $\Gamma\vdash e\leq
e:\myety{\tau}{\eff_1}{\eff_2}{\eff_3}$ can be derived from the rules
from Figure~\ref{tycrule}. We do not formulate explicit typing rules
to save space.

The plan is to justify all the rules semantically using a logical relation (Section~\ref{logrel}) and to then conclude their soundness w.r.t.\ typed observational appoximation and equivalence (Section~\ref{obseq}). 

The parallel composition rule states that two programs $e_1$ and $e_2$
can be composed when their internal effects are not conflicting in the
sense that the internal effects of one program appear as environment
interaction effects of the other program. Note the relationship to the
parallel composition rule of the rely-guarantee formalism
\cite{DBLP:journals/logcom/ColemanJ07}. Also note that the effects of
computations $e_1$ and $e_2$ are not required to be independent from
each other as we do in the parallization rule further down.

 The appearance of the
$\sqcup$-operation deserves special mention. It might be, for example,
that $e_1$ modifies $\cloc$ on the way, thus $\wEff{\cloc}\in\eff_1$
but cleans up this modification by eventually restoring the old value
of $\cloc$. This would be reflected by
$\wEff{\cloc}\notin\eff\cup\eff'\cup\eff_2$. In that case, we would
not expect to see $\wEff{\cloc}$ in the end-to-end effect of the
parallel composition and that is precisely what $\sqcup$ achieves.

The rules labelled (Sem) make available all kinds of program
transformations that are valid on the level of the \emph{untyped}
denotational semantics, including commuting conversions for let and
if, fixpoint unrolling, and beta and eta equalities. 

Finally, we have several effect-dependent (in)equalities: the
parallelization rule generalises a similar rule from
\cite{birkedal}. The other ones are concurrent version of analogous
rules for sequential computation that have been analysed in previous
work
\cite{DBLP:conf/aplas/BentonKHB06,benton07ppdp,DBLP:conf/icfp/ThamsborgB11,DBLP:dblp_conf/popl/Benton0N14}
and are at the basis of all kinds of compiler optimizations.  The side
conditions on the effects are rather subtle and much less obvious than
those found in a sequential setting. The parallelization rule is similar to the parallel congruence rule in that it requires the participating computations to mutually tolerate each other. This time, however, since the two computations being compared will do rather different things temporarily they must be oblivious against chaotic access, hence the $(-)^C$ strengthenings in the premise. 

The reason for the appearance of $(-)^C$ in the other rules is similar. The rule for pure lambda hoist seems unusual and will thus be explained in more detail. First, the computation $e_1$ to be hoisted may indeed have side effects $\eff_1$ so long as they are cleaned up by the time $e_1$ completes and the intervening environment does not notice (modelled by the conditions 
$\eff_1\perp\eff$ and final effect $\eff^C=\eff^C\cup\emptyset$). In the conclusion the transient effect $\eff_1$ shows up again, but $(-)^C$-ed since it only appears in different sides. Also in the other rules like commuting etc.\ 
it is the case that the familiar side conditions on applicability only affect the end-to-end effects whereas the transient effects are merely required not to interfere with the environment. 


\begin{figure*}[t]
\vspace{-3mm}
 \[
\infer{\Gamma\vdash \mtrue\leq\mtrue : \valty{\booltype}}{}
\qquad
\infer{\Gamma\vdash \mfalse\leq\mfalse : \valty{\booltype}}{}
\qquad 
\infer{\Gamma\vdash n\leq n : \valty{\inttype}}{}
\qquad
\infer{\Gamma,x:\tau\vdash x\leq x : \valty{\tau}}{}
\qquad 
\infer{\Gamma\vdash v.i\leq v'.i:\valty{\tau_i}}{\Gamma\vdash
v\leq v':\valty{\tau_1\times\tau_2}}
\]
\[
\infer{\Gamma\vdash e_1\leq e_3:\myety{\tau}{\eff_1}{\eff_2}{\eff_3}}{\Gamma\vdash e_1\leq e_2:\myety{\tau}{\eff_1}{\eff_2}{\eff_3}\quad
\Gamma\vdash e_1\leq e_2:\myety{\tau}{\eff_1}{\eff_2}{\eff_3}
}
\qquad
\infer{\Gamma\vdash v\leq v':\myety{\tau}{\eff_1}{\eff_2}{\eff_3}}{\Gamma\vdash v\leq v':\valty{\tau}}
\qquad
\infer{\Gamma\vdash
(v_1,v_2)\leq(v_1',v_2') :\valty{\tau_1\times\tau_2}}
{\Gamma\vdash v_i\leq v_i':\valty{\tau_1}\mbox{ $i=1,2$}}
\]
\[
\infer{\Gamma\vdash v_1\ v_2\leq v_1'\ v_2' : \myety{\tau_2}{\eff_1}{\eff_2}{\eff_3}}
{\Gamma\vdash v_1\leq v_1':\valty{\tau_1\myeffto{\eff_1}{\eff_2}{\eff_3}\tau_2} & \Gamma\vdash
v_2\leq v_2':\valty{\tau_1}}\quad 
 \infer{\Gamma\vdash
\myif{v}{e_1}{e_2}\leq \ \myif{v'}{e_1'}{e_2'}
   : \myety{\tau}{\eff_1}{\eff_2}{\eff_3}}
{\begin{array}{c}
\Gamma\vdash v\leq v':\valty{\booltype} \\ 
    \Gamma\vdash e_1\leq e_1':\myety{\tau}{\eff_1}{\eff_2}{\eff_3}\ \ 
    \Gamma\vdash e_2\leq e_2':\myety{\tau}{\eff_1}{\eff_2}{\eff_3}
  \end{array}
}
\]\[
\infer{\Gamma\vdash \letin{x}{e_1}{e_2}\leq \letin{x}{e_1'}{e_2'} :\myety{\tau_2}{\eff_1}{\eff_2}{\eff_3}}
{\begin{array}{c}\Gamma\vdash 
e_1\leq e_1':\myety{\tau_1}{\eff_1}{\eff_2}{\eff_3}\\
  \Gamma,x{:}\tau_1\vdash e_2\leq e_2':\myety{\tau_2}{\eff_1}{\eff_2}{\eff_3}
  \end{array}
}
\quad
\infer{\Gamma\vdash \vfix{f}{x}{e}\leq\vfix{f}{x}{e'}  : \valty{\tau_1\myeffto{\eff_1}{\eff_2}{\eff_3} \tau_2}}
{\Gamma,f{:}\tau_1\myeffto{\eff_1}{\eff_2}{\eff_3} \tau_2,x{:}\tau_1\vdash
e\leq e':\myety{\tau_2}{\eff_1}{\eff_2}{\eff_3}}
\]

\[
 \infer{\Gamma \vdash \mypar{e_1}{e_2} \leq \mypar{e_1'}{e_2'} : \myety{\tau_1 \times \tau_2}{\eff_1 \cup \eff_2}{\eff}{\eff \cup \eff' \cup (\eff_1 \sqcup \eff_2)}}
{\Gamma \vdash e_1 \leq e_1' : 
  \myety{\tau_1}{\eff_1}{\eff \cup \eff_2}{\eff \cup \eff_2 \cup \eff'}\quad \Gamma \vdash e_2 \leq e_2' : 
  \myety{\tau_2}{\eff_2}{\eff \cup \eff_1}{\eff \cup \eff_1 \cup \eff'}}
 \]

 \[
  \infer[\textrm{Sem}_1]{\Gamma \vdash e' \leq e' : \myety{\tau}{\eff_1}{\eff_2}{\eff_3}}{\Gamma \vdash e \leq e : \myety{\tau}{\eff_1}{\eff_2}{\eff_3} \quad \sem{e} = \sem{e'}}
  \qquad
  \infer[\textrm{Sem}_2]{\Gamma \vdash e \leq e' : \myety{\tau}{\eff_1}{\eff_2}{\eff_3}}{\Gamma \vdash e \leq e : \myety{\tau}{\eff_1}{\eff_2}{\eff_3} \quad \sem{e} = \sem{e'}}
\qquad 
\infer[\textrm{Ax}_1]{\Gamma\vdash v\leq v:\tau}{(v,v',\tau)\mbox{ an axiom}}
\]

\[ \infer{\Gamma\vdash e\leq e':\myety{\tau}{\eff_1'}{\eff_2'}{\eff_3'}}{\Gamma\vdash
  e\leq e':\myety{\tau}{\eff_1}{\eff_2}{\eff_3} & \eff_1
  \subseteq  \eff_1' & \eff_2' \subseteq  \eff_2 & \eff_3
  \subseteq  \eff_3'}
\quad 
\infer[\textrm{Atom}]{\Gamma\vdash\myatomic{e}\leq\myatomic{e'}:\myety{\tau}{\eff_3}{\eff_2}{\eff_2\cup\eff_3}}{\Gamma\vdash e\leq e':\myety{\tau}{\eff_1}{\emptyset}{\eff_3}} \quad 
\infer[\textrm{Ax}_2]{\Gamma\vdash v'\leq v':\tau}{(v,v',\tau)\mbox{ an axiom}}
\]
\vspace{-2mm}
\label{tycrule}
\caption{Typing and congruence rules}
\vspace{-4mm}
\end{figure*}
\begin{figure*}
\[
\infer[\textrm{Parallelization}]{\Gamma\vdash \mypar{e_1}{e_2} \leq
(\letin{x}{e_1}{\letin{y}{e_2} (x,y)}) : 
\myety{\tau_1 \times \tau_2}
{\eff_1^C \cup \eff_2^C}
{\eff}{ \eff \cup \eff_1' \cup \eff_2'}}
{\begin{array}{l}
 \Gamma\vdash e_1:\myety{\tau_1}
{\eff_1}{\eff^C \cup \eff_2^C}{\eff^C \cup \eff_2^C \cup \eff_1'} \qquad \Gamma\vdash e_2:\myety{\tau_2}
{\eff_2}{\eff^C \cup \eff_1^C}{\eff^C \cup \eff_1^C \cup \eff_2'} \qquad \qquad 
\eff_1 \perp \eff_2 \quad \eff_1 \perp \eff \quad \eff_2 \perp \eff
  \end{array}
}
\]

\[
\infer[\textrm{Commuting}]{
\Gamma\vdash (\letin{x}{e_1}{\letin{y}{e_2} (x,y)}) = (\letin{y}{e_2}{\letin{x}{e_1} (x,y)}) : 
\myety{\tau_1 \times \tau_2}
{\eff_1^C \cup \eff_2^C} {\eff}{ \eff \cup \eff_1' \cup \eff_2'}}
{\begin{array}{l}
 \Gamma\vdash e_1:\myety{\tau_1} {\eff_1}{\eff^C}{\eff^C \cup \eff_1'} \qquad 
\Gamma\vdash e_2:\myety{\tau_2}
{\eff_2}{\eff^C}{\eff^C \cup \eff_2'} \qquad \qquad
\eff_1' \perp \eff_2' \quad \eff_1 \perp \eff \quad \eff_2 \perp \eff
  \end{array}
}
\]

\[
\begin{array}{c}
 \infer[\textrm{Duplicated}]{\Gamma\vdash
(\letin{x}{e} (x,x)) \leq
(\letin{x}{e} {\letin{y}{e} (x,y))}) : \myety{\tau\times \tau} {\eff_1^C} {\eff_2}{\eff_2 \cup \eff'}}
{\Gamma\vdash e:\myety{\tau} {\eff_1}{\eff_2^C}{\eff_2^C \cup \eff'} 
 \quad
\rdsin{\eff'} \cap \wrsin{\eff'} = \emptyset 
\quad \eff_2 \perp \eff_1 
}
\end{array}
\]

\[
\infer[\textrm{Ax}]{\Gamma\vdash v\leq v':\tau}{(v,v',\tau)\mbox{ an axiom}}\quad 
\begin{array}{c}
\infer[\textrm{Lambda Hoist}]{
 \Gamma\vdash \letin{y}{e_1} {\lambda x. e_2} 
\leq \lambda x . \letin{y}{e_1}{e_2} 
: 
\myety{\tau_3 \myeffto{\eff_1^C\cup\eff_2}{\eff}{\eff\cup\eff_3} \tau_2}{\eff_1^C}{\eff}{\eff}}
{\Gamma \vdash e_1: \myety{\tau_1} {\eff_1}{\eff^C}{\eff^C} 
 \qquad
 \Gamma, x:\tau_3, y:\tau_1 \vdash e_2: \myety{\tau_2} {\eff_2}{\eff}{\eff\cup\eff_2}\qquad \eff\perp\eff_1}
\end{array}
\]
\[
\begin{array}{c}
\infer[\textrm{Deadcode}]{
\Gamma\vdash e_2 \leq (\letin{x}{e_1}{e_2}) : \myety{\tau_2}{\eff_1^C\cup\eff_2}{\eff}{\eff \cup \eff_2'} }
{\Gamma \vdash e_1: \myety{\tau_1} {\eff_1}{\eff^C}{\eff^C \cup \eff_1'} 
 \qquad
 \Gamma \vdash e_2: \myety{\tau_2} {\eff_2}{\eff}{\eff_2'} 
 \qquad 
 \eff_1 \perp \eff
 \qquad 
 \textrm{$\wrsin{\eff_1'}=\emptyset$}
}
\end{array}  
\]

\caption{Effect-dependent transformations\label{eqth}.}
\vspace{-3mm}
\end{figure*}
The following definitions provide the semantics of our effect annotations. 
\begin{definition}[Tiling]
Let $\w\vdash\eff$. We write $[\eff](\h,\h',\h_1,\h_1')$ to mean that (i) $\h\models\w\Rightarrow \gloc{\h}{\h_1}{\eff}$ and 
(ii) $\h'\models \w\Rightarrow \gloc{\h'}{\h_1'}{\eff}$ and (iii) $\rloc{\h}{\h'}{\rdsin{\eff}}$ and  $\loc\in\wrs(\eff)\setminus\cos(\eff)$ imply $(\rrloc{\h}{\h_1}{\loc}\wedge\rrloc{\h'}{\h_1'}{\loc})\vee \rloc{\h_1}{\h_1'}{\loc}$.
\end{definition}
Thus, assuming semantic consistency of heaps, $\h$ and $\h'$ evolve to $\h_1$ and $\h_1'$ according to the modifying (writing or chaotic) locations in $\eff$, and if $\h,\h'$ agree on the reads of $\eff$ then written locations will either be identicallly modified or left alone.  

If the step relations of all abstract locations commute with each
other then tiling admits an alternative characterisation in terms of
preservation of binary relations \cite{DBLP:conf/aplas/BentonKHB06}. The present more
operational version is inspired by the treatment of effects in
\cite{birkedal}.

\begin{lemma}\label{tillem}
Suppose that $\w\vdash\eff$, $\w\vdash\eff_1$, $\w\vdash\eff_2$. 
The following hold whenever well-formed. 
\begin{compactenum}
\item \label{tiltrans}
If $[\eff](\h,\h',\h_1,\h_1')$ and $[\eff](\h_1,\h_1',\h_2,\h_2')$ then $[\eff](\h,\h',\h_2,\h_2')$;  
\item $[\eff](\h,\h',\h,\h')$ 
\item\label{tilmon} If $\eff_1\subseteq \eff_2$ then $[\eff_1](\h,\h',\h_1,\h_1')\Rightarrow 
[\eff_2](\h,\h',\h_1,\h_1')$
\item\label{tilmcon} $[\eff](\h,\h',\h_1,\h_1')\Rightarrow 
[\eff^C](\h,\h',\h_1,\h_1')$
\item \label{tilrd} If $[\eff](\h,\h',\k,\k')$ and $\rloc{\h}{\h'}{\rdsin{\eff}}$ then  $\rloc{\k}{\k'}{\rdsin{\eff}}$. (this relies on $\rdsin{\eff}\cap\concs{\eff}=\emptyset$.)
\item\label{tilwf} Suppose $[\eff](\h,\h',\h_1,\h_1')$. If $\h\models\w$ then $\h_1\models\w$; if $\h'\models\w$ then $\h_1'\models\w$. 
\end{compactenum}
\end{lemma}
\section{Logical Relation}\label{logrel}

\begin{definition}[Specifications] A value 
specification is a relation $E\subseteq \Values \times\Values$ such that 
\begin{itemize}
\item if $x_1\leq x$ and $y\leq y_1$ and $x\E y$ then $x_1\E y_1$; 
\item if $(x_i)_i$ and $(y_i)_i$ are chains such that $x_i \E y_i$ then $\sup_i x_i \E \sup_i y_i$, i.e., $E$ is admissible qua relation; 
\item if $x\E y$ then $p_i(x)\E p_i(y)$ for each $i$, i.e.\ $E$ is closed under the canonical deflations. 
\end{itemize}
Similarly,  a computation specification is a 
relation $Q\subseteq T\Values\times T\Values$ such that 
${\leq};Q;{\leq}\subseteq Q$ and $Q$ is admissible qua relation and $Q$ is closed under the canonical deflations $q_i$. 
\end{definition}
The requirement ${\leq};E;{\leq}\subseteq E$ ensures smooth interaction with the down-closure built into our trace monad. Admissibility is needed for the soundness of recursion and closure under the canonical deflations, finally is needed so that Lemma~\ref{func} can be applied. 
\begin{definition}
If $E\subseteq \Values\times\Values$ and $Q\subseteq T\Values\times T\Values$
then 
the relation $E{\rightarrow}Q\subseteq \Values\times\Values$ is defined by 
\[
f E{\rightarrow}Q f' \iff \forall x\ x'.(x \E x') \Rightarrow (f(x)\,Q\, f'(x'))
\] 
In particular, for $f E{\rightarrow}Q f'$ to hold, both $f,f'$ must be functions (and not elements of base type or tuples). 
\end{definition}
\begin{lemma}
If $E$ and $Q$ are specifications so is $E{\rightarrow}Q$. 
\end{lemma}
The following is the crucial definition of this paper; it gives a semantic counterpart to observational approximation and, due to its game-theoretic flavour, allows for very intuitive proofs. 
\begin{definition}
\label{defn:crucial}
Let $E\subseteq \Values\times\Values$ be a value specification 
and $(\eff_1, \eff_2, \eff_3)$ an effect specification. We
define the relations $T_0(E,\eff_1,\eff_2,\eff_3)$ and 
$T(E,\eff_1,\eff_2,\eff_3)$ between sets of trace-value pairs, i.e.\ on $\mathcal{P}(\mathit{Tr}\times\textit{Values})$: 

$(U,U')\in T_0(E,\eff_1,\eff_2,\eff_3)$ if and only if 
\[
\left[\begin{array}{l}
\forall \tup{(\h_1,\k_1) \ldots (\h_n,\k_n),a}\in U.\h_1\models\w\Rightarrow\\
\quad \forall \h_1'. \h_1'\models\w\Rightarrow
\rloc{\h_1}{\h_1'}{\rds(\eff_3)}\Rightarrow\\
\qquad \exists \k_1'. [\eff_1](\h_1,\h_1',\k_1,\k_1') \land
\forall \h_2'. [\eff_2](\k_1,\k_1',\h_2,\h_2') \Rightarrow \\
\quad \qquad \exists \k_2'. [\eff_1](\h_2,\h_2',\k_2,\k_2') \land
\forall \h_3'. [\eff_2](\k_2,\k_2',\h_3,\h_3') \Rightarrow \\
\qquad \qquad \cdots\\
\qquad \qquad\exists \k_n'. [\eff_1](\h_n,\k_n,\h_n',\k_n') \land~
[\eff_3](\h_1,\h_1',\k_n,\k_n') \land\\ 
\qquad \qquad \exists a' \in\Values. 
(a,a')\in E \land  
\tup{(\h_1',\k_1')\ldots (\h_n',\k_n'),a'}\in U'
 \end{array}\right] 
 \]
We define the relation $T(E,\eff_1,\eff_2,\eff_3)\subseteq T\Values\times T\Values$ as the admissible closure of $T_0$, i.e.\ $\textit{Adm}(T_0(E,\eff_1,\eff_2,\eff_3))$. 
\end{definition}
The game-theoretic view of $T_0(E,\eff_1,\eff_2,\eff_3)$ may be
understood as follows. Given $U,U'\in T\Values$ we can consider a
game between a proponent (who believes $(U,U')\in T\Values$)
and an opponent who believes otherwise. The game begins by the
opponent selecting an element $\tup{(\h_1,\k_1) \ldots
  (\h_n,\k_n),a}\in U$ and $\h_1\models\w$, the \emph{pilot trace} and a 
start heap $\h_1'\models\w$ such that $\rloc{\h_1}{\h_1'}{\rdsin{\eff_3}}$
 to begin a trace in $U'$. Then, the proponent answers with a matching heap $\k_1'$ so that $[\eff_1](\h_1,\h_1',\k_1,\k_1')$. If $\rloc{\h_1}{\h_1'}{\rdsin{\eff_1}}$ does not hold, proponent does not need to ensure that writes are in sync. 
The opponent
then plays a heap $\h_2'$ so that
$[\eff_2](\k_1,\k_1',\h_2,\h_2')$. At this point, it is in the proponents interest to make sure that $\rloc{\k_1}{\k_1'}{\rdsin{\eff_2}}$ for otherwise opponent may make ``funny'' moves.

 Then, again, proponent plays a heap
$\k_2'$ such that $[\eff_1](\h_2,\h_2',\k_2,\k_2')$ and so on until,
proponent has played $\k_n'$ so that
$[\eff_1](\h_{n},\h_n',\k_n,\k_n')$. After that final heap has been played, it is checked that $[\eff_3](\h,\h',\k_n,\k_n')$ holds. If not, proponent loses. If yes, then  proponent must also
play a value $a'$ and it is then checked whether or not
$\tup{(\h_1',\k_1')\ldots (\h_n',\k_n'),a'}\in U'$ and $(a \E a')$. If
this is the case or if at any one point in the game the opponent was
unable to move because there exists no appropriate heap then the
proponent has won the game. Otherwise the opponent wins and we have
$(U,U')\in T_0(E,\eff_1,\eff_2,\eff_3)$ iff the proponent has a
winning strategy for that game.

We notice that by Lemma~\ref{tillem}(\ref{tilwf}) well-formedness of
heaps w.r.t.\ the ambient world is a global invariant which allows us
to refrain form explicitly assuming and asserting it in subsequent
proofs and statements.

We now illustrate the game with a few examples. 
\begin{example}
\label{jife}
\normalfont
 Consider the following programs:
 \(
 e_1 = (\assign{\cloc}{\myread{\cloc} + 1}; \assign{\cloc}{\myread{\cloc}+ 1})  \quad \textrm{and} \qquad  e_2 = 
 (\assign{\cloc}{\myread{\cloc} + 2}). 
 \)

 \noindent
Let $\loc = \locInt(\cloc)$ be the abstract location for a single integer stored at $\cloc$ (see Section~\ref{sec:abs-examples}). Let $E=\sem{\unittype}=
\{((),())\}$ be the value specification for the unit type.  

We  show that  $(\sem{e_1}, \sem{e_2}) \in T(E, 
\{\cEff{\loc}\}, \eff, \eff \cup \{\rEff{\loc}, \wEff{\loc}\}\}$ under the assumption that $\orth{\{\cEff{\loc}\}}{\eff}$, that is, when the environment does not 
read nor write $\cloc$.  This condition is clearly necessary, for $e_1$ and $e_2$ can be distinguished by an environment allowed to read or write $\cloc$. 

Let us now prove the claim when $\orth{\{\cEff{\loc}\}}{\eff}$. The opponent picks a pilot trace in the semantics of $e_1$, for example, 
\(
 \tup{(\h_1,\k_1)(\h_2,\k_2),()}
\)
\noindent
where $\h_1(\cloc)=n$ and $\k_1(\cloc)=n+1$ and $\h_2(\cloc)=n'$ and
$\k_2(\cloc)=n'+1$. The other possible traces are stuttering or
mumbling variants of this one and do not present additional
difficulties. The opponent also chooses a heap $\h_1'$ such that
$\rloc{\h_1}{\h_1'}{\loc}$, i.e., $\h_1'(\cloc)=n$.  Now the proponent
will choose to stutter for the time being and thus selects
$\k_1':=\h_1'$. Indeed, $[\cEff{\loc}](\h_1,\h_1',\k_1,\k_1')$ holds,
so this is legal. The opponent now presents $\h_2'$ such that
$[\eff](\k_1,\k_1',\h_2,\h_2')$. By the assumption on $\eff$ we know
that $n'=\h_2(\cloc)=\k_1(\cloc)=n+1$ and also
$\h_2'(\cloc)=\k_1'(\cloc)=n$. The proponent now answers with
$\k_2':=\h_2'[\cloc{\mapsto}n+2]$. It follows that
$[\cEff{\loc}](\h_2,\h_2',\k_2,\k_2')$ and also
$[\rEff{\loc},\wEff{\loc}](\h_1,\h_1',\k_2,\k_2')$. Finally, by
stuttering $(\h_1',\h_1')(\h_2',\h_2'[\cloc{\mapsto}n+2])\in\sem{e_2}$
so that proponent wins the game.
\end{example}
\begin{example}
\label{ex:parallel}
\normalfont
  Consider the following programs $e_1$ and $e_2$:

   \(
 (\mypar{\assign{\cloc}{\myread{\cloc} + 1}}{\assign{Y}{\myread{Y} + 1}})   \quad \textrm{and} \quad  (\assign{\cloc}{\myread{\cloc} + 1}; \assign{Y}{\myread{Y}+ 1}). 
 \)

\noindent
We show $(\sem{e_1},\sem{e_2}) \in T(E, \{\cEff{X},\cEff{Y}\}, \eff, \eff \cup \{\rEff{X},\rEff{Y},\wEff{X},\wEff{Y}\})$, provided $\eff$ does not read nor modify $X$ and $Y$. This equivalence could be deduced syntactically using our parallelization equation shown in Figure~\ref{eqth}. For illustrative purpose, however, we describe its semantic proof using a game. 

The opponent picks a pilot trace in $\sem{e_1}$, for example, the trace $([n_1|n_2],[n_1|n_2+1])([n_1|n_2+1],[n_1+1|n_2+1]) (\unitval,\unitval)$, where $[n_X|n_Y]$ denotes a heap where $X$ and $Y$ store $n_X$ and $n_Y$, respectively. Notice that in this trace, $Y$ is incremented before $X$ and since $\eff$ does not read nor modify $X$ and $Y$, the environment move does not change the values in $X$ nor $Y$. We are also given an initial heap  $\h_1'$ that agrees with the initial heap $[n_1|n_2]$ on the reads of $\eff \cup \{\rEff{X},\rEff{Y},\wEff{X},\wEff{Y}\}$. Thus, $\h_1'$ should be of the form $[n_1|n_2]$. 

We now play the move $([n_1|n_2], [n_1+1|n_2])$. This is a valid move in the game as $[\cEff{X},\cEff{Y}]([n_1|n_2],[n_1|n_2],[n_1|n_2+1],[n_1+1|n_2])$. The environment moves returning $[n_1+1|n_2]$ as it does not read nor modify $X$ and $Y$. We can now match the trace above by playing $([n_1+1|n_2],[n_1+1|n_2+1])$ and returning $(\unitval,\unitval)$, winnning the game.
\end{example}
The following is one of the main technical result of our paper and shows that
the computation specifications $T(\dots)$ can indeed serve as the
basis for a logical relation. We just show here the soundness proof for the parallel congruence rule. The missing proofs appear in the attached Appendix.
\begin{theorem}\label{main}
The following hold whenever well-formed. 
\begin{compactenum}
\item\label{eins} If $(U,U')\in T(E,\eff_1,\eff_2,\eff_3)$ then $(q_i(U),q_i(U'))\in T(E,\eff_1,\eff_2)$. 
\item\label{einsa} $T(E,\eff_1,\eff_2,\eff_3)$ is a computation specification. 
\item\label{zwei} If $(U,U')\in T(E,\eff_1,\eff_2,\eff_3)$ then $(U^\dagger,{U'}^\dagger)\in T(E,\eff_1,\eff_2,\eff_3)$. 
\item\label{drei} If $(a,a')\in E$ then $(\textit{rtn}(a),\textit{rtn}(a'))$ is in $T(E,\eff_1,\eff_2,\eff_3)$. 
\item\label{vier} Suppose that $(\eff_1,\eff_2,\eff_3)$ is an effect specification where $\eff_1\cup\eff_2\subseteq \eff_3$. Suppose that whenever $\rloc{\h}{\h'}{\rdsin{\eff_1}}$ and $c(\h)=(\h_1,a)$ then there exist $(\h_1',a')$ such that $c'(\h')=(\h_1',a')$ and $[\eff_1](\h,\h',\h_1,\h_1')$ and $aEa'$. We then have  for any $\eff_2$,  $(\textit{fromstate}(c),\textit{fromstate}(c'))\in T(E,\eff_1,\eff_2,\eff_3)$. 
\item\label{fuenf} If 
 $(f,f')\in E_1{\rightarrow} T(E_2,\eff_1,\eff_2,\eff_3)$ and $(U,U')\in T(E_1,\eff_1,\eff_2,\eff_3)$ then $(\textit{bnd}(f,U),\textit{bnd}(f',U'))\in T(E_2,\eff_1,\eff_2,\eff_3)$. 
\item\label{sechs} If $(U_1,U_1')\in T(E_1,\eff_1,\eff\cup\eff_2,\eff\cup\eff_2\cup\eff')$ and $(U_2,U_2')\in T(E_2,\eff_2,\eff\cup\eff_1,\eff\cup\eff_1\cup\eff')$ then $(U_1\semparallel U_1',U_2\semparallel U_2')\in 
T(E_1\times E_2,\eff_1\cup\eff_2,\eff,\eff\cup\eff'\cup(\eff_1\sqcup\eff_2))$.
\item\label{acht} $(U,U')\in T(E,\eff_1,\emptyset,\eff_3)$ $\Rightarrow$ 
$(\textit{at}(U), \textit{at}(U'))\in T(\eff_3,\eff_2,\eff_2\cup\eff_3)$.  
\end{compactenum}
\end{theorem}
\begin{proof}
Ad \ref{sechs}. Suppose that $(U_1,U_1')\in
T(E_1,\eff_1,\eff\cup\eff_2,\eff\cup\eff_2\cup\eff')$ and $(U_2,U_2')\in T(E_2,\eff_2,\eff\cup\eff_1,\eff\cup\eff_1\cup\eff')$ and let $(t,(a,b))\in U_1\semparallel U_2$,
thus $\textit{inter}(t_1,t_2,t)$ (ignoring $\dagger$ by item
\ref{zwei}) where $(t_1,a)\in U_1$ and $(t_2,b)\in U_2$. Let $S_1$,
$S_2$ be corresponding winning strategies.  The idea is to use $S_1$
when we are in $t_1$ and to use $S_2$ when we are in $t_2$. Supposing
that $t$ starts with a $t_1$ fragment we begin by playing according to $S_1$. Let $t$ be of the form:
\[
\begin{array}{ll}
  t = & (\h_1,\k_1) \cdots (\h_n,\k_n) (\h_{n+1},\k_{n+1}) \cdots (\h_{n+m},\k_{n+m})\\ & (\h_{n+m+1},\k_{n+m+1}) \cdots (\h_{n+m+k},\k_{n+m+k})  \cdots (\h_p,\k_p)  
\end{array}
\]
composed of pieces of the traces $t_1$ and $t_2$. Assume w.l.o.g. that the first piece $(\h_1,\k_1) \cdots (\h_n,\k_n)$ is a part of $t_1$. We are given a initial heap $\h_1'$ such that $\rloc{\h}{\h'}{\rds(\eff\cup\eff'\cup(\eff_1\sqcup\eff_2))}$. Since $\rds(\eff_1\sqcup\eff_2)=\rds(\eff_1)\cup\rds(\eff_2)$, we can apply strategy $S_1$ to guide us through the first part of the game, obtaining:
\[
  (\h_1',\k_1') \cdots (\h_n',\k_n')
\]
Moreover, we have an environment move which forms the tile $[\eff](\k_n,\k_n',\h_{n+1},\h_{n'+1})$. Thus, we have the tile $[\eff \cup \eff_1](\h_1,\h_1',\h_{n+1},\h_{n+1}')$ which can be seen as an environment move for $t_2$. Therefore, we can use strategy $S_2$ for the $U'$ and continue the game, obtaining the trace piece:
\[
  (\h_{n+1}',\k_{n+1}') \cdots (\h_{n+m}',\k_{n+m}')
\]
Now, we can return to the $S_1$ game as the trace above is seen as an environment move for $U$. Alternating these strategies, we get a trace $t$ which is in $(U \semparallel U')$. Let $(a',b')$ be the final values reached at the end. It is clear that $[\eff\cup\eff'\cup\eff_1\cup\eff_2](\h,\h',\h_p,\h_p')$ and also $aE_1a'$ and $bE_2b'$. 

It remains to assert the stronger statement  $[\eff\cup\eff'\cup(\eff_1\sqcup\eff_2)](\h,\h',\h_p,\h_p')$. To see this suppose that $\wEff\loc\in\eff_1\setminus \eff_2\setminus\eff\setminus\eff'$. 
Since the entire game can be viewed as an instance of the game $U_1$ vs $U_1'$ with interventions by $U_2$ vs.\ $U_2'$ regarded as environment interactions we have $[\eff\cup\eff_2\cup\eff'](\h,\h',\h_p,\h_p')$ so that in fact 
$\rrloc{\h}{\h_p}{\loc}$ and $\rrloc{\h'}{\h_p'}{\loc}$. The case of $\cEff{\loc}$ and $\eff_1$,$\eff_2$ interchanged is analogous.
\end{proof}
We assign a value specification $\sem{\tau}$ 
to each refined type by
\[
\begin{array}{l}
 \textrm{\textbullet\quad} \sem{\inttype} =  \{(v,v')\mid v=v'\in
\mathbb{Z}\} \quad
\textrm{\textbullet\quad} \sem{\tau_1\times\tau_2} = \sem{\tau_1}\times\sem{\tau_2}\\
\textrm{\textbullet\quad} \sem{\tau_1\myeffto{\eff_1}{\eff_2}{\eff_3} \tau_2} = \sem{\tau_1}{\rightarrow}T(\sem{\tau_2},\eff_1,\eff_2,\eff_3)
\end{array}\]
We omit the obvious definition of the other basic types and assume
value specifications for user-specified types as given.

\begin{assumption}\label{assi}
We henceforth adopt the following \emph{soundness assumption} which must be established concretely for every concrete instance of our framework. 
\begin{itemize}
\item The initial heap satisfies the current world: $\hinit\models\w$. 
\item Each axiom is type sound: whenever $(v,v',\tau)$ is an axiom then 
 $(v,v)\in\sem{\tau}$ and $(v',v')\in\sem{\tau}$. 
\item Each axiom is inequationally 
sound: whenever $(v,v',\tau)$ is an axiom then 
 $(v,v')\in\sem{\tau}$. 
\end{itemize}
\end{assumption}
\begin{theorem}\label{tysound}
Suppose that $\Gamma\vdash v:\tau$ and $\Gamma\vdash e:
\myety{\tau}{\eff_1}{\eff_2}{\eff_3}$. Then
$(\eta,\eta')\in\sem{\Gamma}$ (interpreting a context as a cartesian
product) implies $(\semV{v}\eta,\semV{v}\eta')\in\sem{\tau}$ and
$(\semC{e}\eta,\semC{e}\eta')\in T(\sem{\tau},\eff_1,\eff_2,\eff_3)$.
\end{theorem}
\begin{proof}
By induction on derivations. Most cases are already subsumed by
Theorem~\ref{main}.
The typing rules regarding functions and recursion follow from the definitions and from the fact that all specifications are admissible.  
\end{proof}
\section{Typed observational approximation}\label{obseq}

\begin{definition}[Observational approximation]
Let $v,v'$ be value expressions where $\vdash v:\tau$ and $\vdash
v':\tau$. We say that $v$ observationally approximates $v'$ at type
$\tau$ if for all $f$ such that $\vdash
f:\tau\myeffto{\eff_1}{\eff}{\eff_3} \inttype$ (``observations'') it
is the case that if $((\hinit,\k),n)\in\sem{f\ v}$ for
$v\in\mathbb{Z}$ and starting from $\hinit$ then $((\hinit,\k'),n)\in
\sem{f\ v'}$ for some $\k'$. We write $\vdash v\leq_{\mathit{obs}}v'$
in this case. We say that $v$ and $v'$ are observationally equivalent
at type $\tau$, written $\vdash v =_{\mathit{obs}}v'$ if both $\vdash
v\leq_{\mathit{obs}}v':\tau$ and $\vdash v'\leq_{\mathit{obs}}v:\tau$.
\end{definition} 
This means that for every test harness $f$ we build around $v$ and
$v'$, no matter how complicated it is and whatever environments it
sets up to run concurrently with $v$ and $v'$ it is the case that each
terminating computation of $v$ (in the environment installed by $f$)
can be matched by a terminating computation with the same result by
$v'$ in the same environment. It is important, however, that the
environment be well typed, thus will respect the contracts set up by
the type $\tau$. E.g.\ if $\tau$ is a functional type expecting, say,
a pure function as argument then, by the typing restriction, the
environment $f$ cannot suddenly feed $v$ and $v'$ a side-effecting
function as input.

We remark that observational approximation extends canonically to open
terms by lambda abstracting free variables (and adding a dummy
abstraction in the case of closed terms)
\cite{DBLP:dblp_conf/popl/Benton0N14}.

As usual, the logical relation is sound with respect to typed
observational approximation and thus can be used to deduce nontrivial
observational approximation relations. We state and prove the precise
formulation of this result.

\begin{theorem}
Let $v,v'$ be closed values and suppose that $(\sem{v},\sem{v'})\in\sem{\tau}$. 
Then $\vdash v\leq_{\mathit{obs}}v':\tau$. 
\end{theorem}

\begin{proof}
If $\vdash f:\tau\myeffto{\eff_1}{\eff_2}{\eff_3}\inttype$ then by Thm~\ref{tysound} we have $(\sem{f},\sem{f})\in \sem{\tau\myeffto{\eff_1}{\eff_2}{\eff_3} \inttype}$, so 
$(\sem{f\ v},\sem{f\ v'})\in T(\sem{\inttype},\eff_1,\eff_2,\eff_3)^+$. 

Let $((\hinit,\k),v)\in \sem{f\ v}$. We have $\hinit\models\w$ and
thus in particular
$\rloc{\hinit}{\hinit}{\rdsin{\eff_3}\cup\rdsin{\eff_1}}$. There must
therefore exist a matching heap $\k'$ and a value $v'$ such that
$((\hinit,\k'),v')\in \sem{f\ v'}$ and $v=v'\in\mathbb{Z}$.
\end{proof}
This means that the examples from earlier on give rise to valid transformations in the sense of observational approximation. For instance, for $e_1$ and $e_2$ form Example~\ref{jife} we find that 
$\lambda \_.e_1 =_{\mathit{obs}}  \lambda\_.e_2$ at type $\unittype\myeffto{\{\cEff{\loc}\}}{\eff}{\eff \cup \{\rEff{\loc}, \wEff{\loc}}\unittype$ whenever $\cloc$ does not appear in $\eff$. 

\section{Effect-dependent transformations}
We will now establish the semantic soundness of the inequational
theory of effect-dependent program transformations given in
Figure~\ref{eqth}.  It includes concurrent versions of the
effect-dependent equations from \cite{DBLP:conf/aplas/BentonKHB06,DBLP:conf/icfp/ThamsborgB11},
but the side conditions on the environmental interaction are by no
means obvious. We also note that some equations now only hold in one
direction thus become inequations. This is in particular the case for
duplicated computations. Suppose that $\mnond{}$ is a computation that
nondeterministically chooses a boolean value and let
$e:=\letin{x}{\mnond{}}{(x,x)}$. Then, even though $\mnond{}$ does not read
nor write any location we only have $e\leq (\mnond{},\mnond{})$, but not
$(\mnond{},\mnond{})\leq e$ for $(\mnond{},\mnond{})$ admits the result
$(\mtrue,\mfalse)$ but $e$ does not. Furthermore, due to presence of
nontermination the equations for dead code elimination and pure lambda hoist also hold in one direction only. It might be possible to restore both directions of said equations by introducing special effects for nondeterminism and nontermination; we have not explored this avenue. We concentrate the individual effect-dependent transformations before
summarising the foregoing results in the general soundness Theorem
\ref{eqthm}. 

In many of the equations, co-effects play an important role. For example, in the commuting and parallelization equations, the internal effects $\eff_1$ and $\eff_2$ in the premises are replaced by  $\eff_1^C$ and $\eff_2^C$ in the internal effects of the conclusion. This makes sense intuitively because the computations are run in a different order, so for the internal moves, the locations in $\eff_1$ and $\eff_2$ can be modified in any way (see Example~\ref{ex:parallel}). However, in the global effect, we can still guarantee the effects $\eff_1'$ and $\eff_2'$ because of the $\perp$-conditions. This intuition appears directly in the soundness proofs.

The following thus constitutes the second main technical
result of our paper. We sketch the soundness proof for parallelization. The detailed proofs appear in the attached Appendix.
\begin{theorem}\label{mainzwei}
The following hold whenever well-formed. 
\begin{itemize}
\item \textbf{Commuting} \label{commusound} If $(U_1,U'_1)\in T(E_1,\eff_1,\eff^C,\eff^C \cup \eff_1')$ and 
$(U_2,U'_2)\in T(E_2,\eff_2,\eff^C,\eff^C \cup \eff_2')$ and $\eff_1\perp\eff$ and $\eff_2\perp\eff$ and $\eff_1'\perp\eff_2'$ then
\begin{small}
\[  
\begin{array}{c}
 (\{(t_1t_2,(v_1,v_2))\mid (t_1,v_1)\in U_1, (t_2,v_2)\in U_2\}^\dagger,\\ 
 \{(t_2't_1',(v_1',v_2'))\mid (t_1',v_1')\in U_1', (t_2',v_2')\in U_2'\}^\dagger)\\
 \in T(E_1\times E_2,(\eff_1\cup\eff_2)^C,\eff,\eff\cup \eff_1'\cup\eff_2')
\end{array}
\] 
\end{small}%
\item\textbf{Duplicated} \label{dupsound} 
If $(U,U') \in T(E,\eff_1,\eff_2^C,\eff_2^C \cup \eff')$ and $\rdsin{\eff'} \cap \wrsin{\eff'} = \emptyset$ and $\eff_2 \perp \eff_1$, then
\[
\begin{array}{c}
 (\{(t,(v,v))\mid (t,v)\in U\}^\dagger,
\{(t_1't_2',(v_1',v_2'))\mid 
(t_1',v_1') \in U', \\
 (t_2',v_2') \in U'\}^\dagger) \in T(E,\eff_1,\eff_2,\eff_2 \cup \eff')  
\end{array}
\]
\item \label{puresound} \textbf{Pure} Let $(U,U') \in T(E,\eff_1, \eff_2^C,\eff_2^C)$, such that $\eff_1\perp\eff_2$. If $((q_1,k_1)\dots(q_n,k_n),v) \in U$ for some \emph{arbitrary} trace $t=(q_1,k_1)\dots(q_n,k_n)$ (with $q_1\models \w$) and  value $v$, then $(\textit{rtn}(v),U') \in T(E,\eff_1^C,\eff_2,\eff_2)$;

\item\label{deadsound} \textbf{Dead} Suppose that $(U,U')\in T(\unittype, \eff_1,\eff_2,\eff_2\cup\eff_1')$, where $\wrsin{\eff_1'}=\emptyset$  and $\eff_1 \perp \eff_2$. Then $(U,\textit{rtn}(\unitval)) \in T(\unittype, \eff_1^C,\eff_2,\eff_2\cup\eff_1')$.



\item\label{parizesound} 
 \textbf{Parallelization} If $(U_1,U_1') \in T(E_1,\eff_1,\eff^C \cup \eff_2^C,\eff^C \cup \eff_2^C \cup \eff_1') $ and $(U_2,U_2') \in T(E_2,\eff_2,\eff^C \cup \eff_1^C,\eff^C \cup \eff_1^C \cup \eff_2')$ and $\eff_1 \perp \eff_2$ and $\eff_1 \perp \eff$ and $\eff_2 \perp \eff$, then 
\[
\begin{array}{c}
(\mypar{U_1}{U_2},\{(t_1't_2'(v_1',v_2'))\mid 
(t_1',v_1') \in U_1', (t_2',v_2') \in U_2'\}^\dagger) \in \\
  T(E_1 \times E_2,\eff_1^C \cup \eff_2^C, \eff, \eff \cup \eff_1' \cup \eff_2')
\end{array}
\]
\end{itemize}
\end{theorem}
\begin{proof}
(Sketch) \textbf{Parallelization.} 

Assume w.l.o.g.\ that the pilot trace takes the form $(t,(v_1,v_2))$ where $\textit{inter}(t_1,t_2,t)$ and $(t_i,v_i)\in U_i$. Just as in the commuting case we set up two side games $U_i$ vs.\ $U_i'$ on $t_i,v_i$. Unlike, in that case, however, these games are running simultaneously and along with the main game. Moves by the environment in the main game are forwarded to the side game we are currently in, i.e., the one to which the current portion of $t$ being played on belongs. At each change of control, we switch between the two side games making last sequence of moves of the other game into a single environment move. It is here that the resilience against chaotic modification is needed. Once the play is over we then assert the claims about the end-to-end effect $\eff\cup\eff_1'\cup\eff_2'$ location by location using the definition of tiling. 
\end{proof}

\begin{theorem}
\label{eqthm}
Suppose that $\Gamma\vdash v\leq v':\tau$ and $\Gamma\vdash e\leq e': \myety{\tau}{\eff_1}{\eff_2}{\eff_3}$ and assume that for each 
 axiom $(v,v',\tau)$ it holds that $(v,v')\in\sem{\tau}^+$.
 Then $(\eta,\eta')\in\sem{\Gamma}^+$ (interpreting a context as a cartesian product) implies $(\semV{v}\eta,\semV{v'}\eta')\in\sem{\tau}^+$ and 
$(\semC{e}\eta,\semC{e'}\eta')\in T(\sem{\tau},\eff_1,\eff_2,\eff_3)^+$. 
\end{theorem}
\begin{proof}[Sketch]
In essence the proof is by induction on derivations of inequalities. However, we need to slightly strengthen the induction hypothesis as follows: 

Define 
\[
\begin{array}{l}

\sem{\Gamma\vdash \tau}=\{(f,f')\mid\forall (\eta,\eta')\in\sem{\Gamma}.(f(\eta),f'(\eta'))\in\sem{\tau}\}\\
\sem{\Gamma\vdash \tau\&(\eff_1,\eff_2,\eff_3)}=\{(f,f')\mid\forall (\eta,\eta')\in\sem{\Gamma}.\\ \qquad \qquad \qquad \qquad \quad (f(\eta),f'(\eta'))\in T(\sem{\tau},\eff_1,\eff_2,\eff_3)\}
\end{array}
\]
We now show by induction on derivations that 
$\Gamma\vdash v\leq v':\tau$ implies $(\sem{v},\sem{v'})\in \sem{\Gamma\vdash \tau}^+$ and that  $\Gamma\vdash e\leq e': \myety{\tau}{\eff_1}{\eff_2}{\eff_3}$  implies 
$(\sem{e},\sem{e'})\in \sem{\Gamma\vdash \tau\&(\eff_1,\eff_2,\eff_3)}^+$. 

The various cases now follow from earlier results in a straightforward manner. Namely, we use Theorem~\ref{main} for the congruence rules and 
Theorem~\ref{mainzwei} for the effect-dependent transformations. 

As a representative case we show the case where $e\equiv\letin{x}{e_1}{e_2}$ and $e'\equiv \letin{x}{e_1'}{e_2'}$. Inductively, we know 
$(\sem{e_1},\sem{e_1'})\in\sem{\Gamma\vdash \tau_1\&(\eff_1,\eff_2,\eff_3)}^{n_1}$ and $(\sem{e_1},\sem{e_1'})\in\sem{\Gamma,x{:}\tau_1\vdash \tau\&(\eff_1,\eff_2,\eff_3)}^{n_2}$ for some $n_1,n_2>0$. By Theorem~\ref{tysound}, we also have
$(\sem{e_1},\sem{e_1})\in\sem{\Gamma\vdash \tau_1\&(\eff_1,\eff_2,\eff_3)}$ and analogous statements for $e_1',e_2,e_2'$. We can, therefore, assume, w.l.o.g.\ that $n_1=n_2$ and then use Theorem~\ref{main} (\ref{fuenf}) repeatedly ($n_1$ times) so as to conclude $(\sem{e},\sem{e})\in\sem{\Gamma\vdash \tau\&(\eff_1,\eff_2,\eff_3)}^{n_1}$. 

The rules for dead code and pure lambda hoist rely on 
 the cases ``Dead'' and ``Pure'' of Thm~\ref{mainzwei} in a slightly indirect way. We sketch the argument for pure lambda hoist. The pilot trace begins with a trace belonging to $e_1$ and yielding a value $v$ for $x$. We can then invoke case ``Pure'' on subsequent occurrences of $e_1$ in the right hand side. 
\end{proof}

\begin{theorem}
Suppose that $\vdash v:\tau$ and $\vdash v':\tau$ and that $(\sem{v},\sem{v'})\in\sem{\tau}^+$ where $(-)^+$ denotes transitive closure. 
Then $\vdash v\leq_{\mathit{obs}} v':\tau$. 
\end{theorem}
\begin{proof}
If $\vdash f:\tau_1\myeffto{\eff_1}{\eff_2}{\eff_3}\inttype$ then by Thm~\ref{tysound} we have $(\sem{f},\sem{f})\in \sem{\tau\myeffto{\eff_1}{\eff_2}{\eff_3} \inttype}$, so 
$(\sem{f\ v},\sem{f\ v'})\in T(\sem{\inttype},\eff_1,\eff_2,\eff_3)^+$. 

Let $((\hinit,\k),v)\in \sem{f\ v}$. We have $\hinit\models\w$ and
thus in particular
$\rloc{\hinit}{\hinit}{{\rdsin{\eff_3}\cup\rdsin{\eff_1}}}$. There must
therefore exist a matching heap $\k'$ and a value $v'$ such that
$((\hinit,\k'),v')\in \sem{f\ v'}$ and $v=v'\in\mathbb{Z}$.
\end{proof}


 
 
 

We now return to the examples that we discussed in Section~\ref{examples} and demonstrate how to prove using our denotational semantics the properties that have been discussed informally.

\paragraph{Overlapping References}
With this example, we illustrate the parallelization rule. In
particular, the functions declared in Section~\ref{examples} have the
following type, where $\eff$ does not read nor write $\cloc$:

\begin{small}
\(
\begin{array}{l}
\mathsf{readFst} : \unittype \myeffto{\emptyset}
            {\eff^C, \cEff{\locHi(\cloc)}}
            {\eff^C, \cEff{\locHi(\cloc)}, \rEff{\locLo(\cloc)}} \inttype\\
\mathsf{writeFst} : \inttype \myeffto{\wEff{\locLo(\cloc)}}
            {\eff^C, \cEff{\locHi(\cloc)}}
            {\eff^C, \cEff{\locHi(\cloc)}, \wEff{\locLo(\cloc)}} \unittype 
\end{array}\)
\end{small}

The obvious and analogous typings for $\mathsf{readSnd}$ and
$\mathsf{writeSnd}$ are elided.  We justify this typing semantically
as described in Theorem~\ref{main}. To illustrate how this is done,
consider the function $(\mathsf{writeSnd}~17)$. We show how the game
is played against itself using the typing shown above. We start with a
``pilot trace'', say:
\(
([2|3],[2|3]),([2|17],[2|17]),(())
\)

\noindent 
where $[x|y]$ denotes a store with $X=p(x,y)$ and other components left out for simplicity.
The first step corresponds to our reading of $X$ and in the second step
-- since there was no environment intervention -- we write $17$ into the first
component.

We now start to play: Say that we start at the heap $[13|12]$. We answer $[13|12]$. If the environment does not change $X$, then we write $17$ to its first component resulting in the following trace, which is possible for $\mathsf{writeFst}(17)$.

\(
([13|12],[13|12]),([13|12],[17|12]),  (())
\)

\noindent
If, however, the environment plays $[18|21]$ (a modification of both
components of X has occurred), then we answer $[17|21]$. Again,

\(
([13|12],[13|12]),([18|21],[17|21]), (())
\)

\noindent
is a possible trace for $\mathsf{writeFst}(17)$. It is easy to check that there is a strategy that justifies the typing given above.

\noindent
Now, consider a program, $e_1$, that only calls $\mathsf{readFst},\mathsf{writeFst}$, and another program, $e_2$, that only calls $\mathsf{readSnd},\mathsf{writeSnd}$. Since the former functions have disjoint effects to the latter ones, $e_1$ and $e_2$ will have effect specifications, respectively, of the form 
$(\eff_1, \eff^C \cup \eff_2^C, \eff^C \cup \eff_2^C \cup \eff_1)$ and $(\eff_2, \eff^C \cup \eff_1^C, \eff^C \cup \eff_1^C \cup \eff_2)$, 
where $\eff_1 \cap \eff_2 = \eff_1 \cap \eff = \eff_2 \cap \eff = \emptyset$. Thus we can use the parallelization rule shown in Figure~\ref{eqth} to conclude that 
the behavior of $\mypar{e_1}{e_2}$ is the same as executing these programs sequentially, although they read and write to the same concrete location.

\paragraph{Loop Parallelization}
We show that the function $\mathsf{map}$ is equivalent to $\mathsf{map2Par}$. It is easy to see that the function $\mathsf{map}$ is equivalent to the program $\mathsf{map2Seq}$, which is the program obtained from $\mathsf{map2Par}$ by replacing the underlined parallel operator `$\|$' in $\mathsf{map2Par}$ by a sequential operator `$;$'. The proof goes simply by unfolding $\mathsf{map}$. 

We then proceed by showing $\mathsf{map2Seq}$ and $\mathsf{map2Par}$ are equivalent using our equations and the abstract locations $\locListOdd(\cloc)$ and $\locListOdd(\cloc)$ defined above. The piece of code that applies $f$ first, namely $e_1 = n.ele := f(n.ele)$, has global effects $\eff_1' = \rEff{\locListOdd(\cloc)},\wEff{\locListOdd(\cloc)}$, while the second application, namely, $e_2 = n.next.ele := f(n.next.ele)$, has effects $\eff_2' = \rEff{\locListEven(\cloc)},\wEff{\locListEven(\cloc)}$. Notice that $\eff_1'\perp\eff_2'$. Therefore, provided that the environment does not read nor modify the list, we can apply the parallelization equation to justify running $e_1$ and $e_2$ parallel is equivalent to running them in sequence.

\paragraph{Michael-Scott Queue}
We now show that the $\mathsf{enqueue}$ and $\mathsf{dequeue}$ functions described in Section~\ref{examples} for the Michael-Scott Queue have the same behavior as their atomic versions. 
We only show the case for $\mathsf{dequeue}$, as the case for $\mathsf{enqueue}$ is similar. More precisely, we now justify the axiom
\[
(\mathsf{dequeue},\myatomic{\mathsf{dequeue}},
\unittype\myeffto{\textit{MSQ}}{\textit{MSQ}}{\textit{MSQ}}\inttype)
\]
\noindent
where $\textit{MSQ}= \{\rEff{\locMSQ(\cloc)},
\wEff{\locMSQ(\cloc)}\}$.  That is, they approximate each other
 at a type where the environment is
allowed to operate on the queue as well. We also note that the
converse of the axiom is obvious by stuttering and mumbling.
After consuming a dummy argument $()$ let the resulting pilot trace be
$(\h_1, \k_1)\ldots(\h_i,\k_i)\ldots(\h_n,\k_n)a$ and $\h_1'$ be the
start heap to match. We can now assume that the passages from $\k_i$
to $\h_{i+1}$ are according to the protocol,
i.e.\ $\gloc{\k_i}{\h_{i+1}}{\locMSQ(\cloc)}$. Namely, should this not
be the case we are free to make arbitrary moves and still win the
game by default of the environment player. Therefore, there must exist
$i$ such that in the move $(\h_i, \k_i)$ the element $a$ is dequeued
and $\h_j=\k_j$ holds for $j\neq i$.  We can thus match this trace by
a trace in the semantics of $\myatomic{\mathsf{dequeue~()}}$ by
stuttering until $i$:

\(
(\h_1', \h_1')\ldots(\h_i',\ldots
\)

\noindent
where $\h_j$ and $\h_j'$ have the same content, but not necessarily the exact same layout. Given the environment's allowed effects it is then clear that also $\h_i$ and $\h_i'$ have the same content, but not necessarily the same as $\h_1$ and $\h_1'$ because in the meantime other operations on the queue might have succeeded. We then dequeue the corresponding element from $\h_i'$ leading to $\k_i'$ and continue by stuttering. 

\(
\ldots,\k_i')(\h_{i+1}',\h_{i+1}')\ldots (\h_n', \h_n')a'
\)

\noindent
It is now clear that this is a matching trace and that $a=a'$ so we are done. 

Notice that the congruence rules now allow us to deduce the equivalence of 
$\textit{op}_1\parallel \cdots \parallel \textit{op}_n$ and 
$\myatomic{\textit{op}_1}\parallel \cdots \parallel \myatomic{\textit{op}_n}$ 
for $\textit{op}_i$ being enqueues or dequeues, which effectively amounts to linearizability. 
\section{Discussion}
We have shown how a simple effect system for stateful computation and
its relational semantics, combined with the notion of abstract
locations, scales to a concurrent setting. The resulting type system
provides a natural and useful degree of control over the otherwise
anarchic possibilities for interference in shared variable languages,
as demonstrated by the fact that we can delineate and prove the
conditions for non-trivial contextual equivalences, including
fine-grained data structures.

The primary goal of this line of work is not so much to find reasoning
principles that support the most subtle equivalence arguments for
particular programs, but rather to capture more generic properties of
modules, expressed in terms of abstract locations and relatively
simple effect annotations, that can be exploited by clients (including
optimizing compilers) in external reasoning and transformations. But
there are of course, particularly in view of the fact that we allow
deeper reasoning to be used to establish that expressions can be
assigned particular effect-refined types, very close connections with
other work on richer program logics and models. 

Rely-guarantee reasoning is widely used in program logics for concurrency, including
relational ones \cite{liangfengpopl12}, whilst our abstract locations
are very like the \emph{islands} of Ahmed et al
\cite{DBLP:conf/popl/AhmedDR09}. Recent work of Turon et al
\cite{dreyer} on relational models for fine-grained concurrency
introduces richer abstractions, notably state transition systems
expressing inter-thread protocols that can involve ownership
transfer. These certainly allow the verification of more complex
fine-grained algorithms than can be dealt with in our setting, and it
would be natural to try defining an effect semantics over such a
model. Indeed, one might reasonably hope that effects could provide
something of a `simplifying lens', with refined types capturing things
that would otherwise be extra model structure or more complex
invariants, such that the combination does not lead to further
complexity. The use of Brookes's trace model (also used
by, for example, Turon and Wand \cite{DBLP:conf/popl/TuronW11}) already seems to
bring some simplification compared to transition systems or
resumptions. 

Birkedal et al \cite{birkedal} have also studied relational semantics
for effects in a concurrent language. The language considered there has dynamic allocation via regions and higher-order store, neither of which we have here. On the other hand, their invariants are based on simply-typed concrete locations and thus do not allow to capture effects at the level of whole datastructures as abstract locations do. As a result, the examples in \cite{birkedal} are of a simpler nature than ours. Furthermore, we offer a subtler parallelization rule, distinguish transient and end-to-end effects, and validate other effect-dependent equivalences like commuting, lambda
hoist, deadcode and duplication. Our  use of denotational methods and in
particular the extension of Brookes' trace semantics to higher-order
functions does result in a rather simpler and more intuitive 
definition of the logical relation by
comparison with \cite{birkedal}. While some of the complications are  due to the
dynamic allocation and typed locations, others like the explicit step
counting, the need for effect-instrumented operational semantics, and the separation of
branches in the definition of safety are not. We thus see our work also as a proof-of-concept for
denotational semantics in the realm of higher-order concurrent
programming.

The `RGSim' relation proposed by Liang \etal\ for proving concurrent
refinements under contextual assumptions also has many similarities
with our logical relation~\cite[Def.4]{liangfengpopl12}. The focus of
that work is on proving particular equivalences and refinements,
whereas we encapsulate general patterns of behaviour in a refined type
system and can show the soundness of generic program transformations
relying only on effect types (which combine smoothly with hand proofs
of particular equivalences).

There are many directions for further work. Most importantly, we would
like to add dynamic allocation of abstract locations following
\cite{DBLP:dblp_conf/popl/Benton0N14}. In addition to relieving us
from having to set up all data structures in the initial heap this
would, as we believe, also allow us to model and reason about
lock-based protocols in an elegant way. Other possible extension
include higher-order store and weak concurrency models.

\newpage

\bibliography{bib}

\newpage 

 
\appendix

\section{Proof of Theorem~\ref{main}}

\begin{proof}
In each case, using Corollary~\ref{funad} and Lemma~\ref{func} (for
case \ref{fuenf}), we can in fact assume w.l.o.g.\ that the assumed
pairs are in $T_0(\dots)$ rather than $T(\dots)$. 

\medskip 

 Ad~\ref{eins}.  Let $(t,a)\in q_i(U)$, i.e.\ $a=p_i(a_0)$ where
 $(t,a_0)\in U$.  By down-closure ([Down]) we also have $(t,a)\in
 U$. We can now play the strategy guaranteed by the assumption
 $(U,U')\in T(E,\eff_1,\eff_2,\eff_3)$ which will yield (depending
 on the opponent's moves) a trace $t'$ and a value $a'$ such that
 $(t',a')\in U'$ and $(p_i(a),a')\in E$. Now, since $E$ is a
 specification we get $(p_i(a),p_i(a'))\in E$ noting that $p_i$ is
 idempotent. So, we modify the strategy so as to return $p_i(a')$
 rather than $a'$ and thus obtain a winning strategy asserting the
 desired conclusion.

Ad~\ref{einsa} This is an easy consequence from \ref{eins}. 

Ad~\ref{zwei} Pick $(U,U')\in
T_0(E,\eff_1,\eff_2,\eff_3)$. Since $T(E,\eff_1,\eff_2,\eff_3)$ is
closed under suprema it suffices to show that
$(q_j(U^\dagger),q_j({U'}^\dagger))\in T(E,\eff_1,\eff_2,\eff_3)$ for
  each $j$. Fix such $j$ and pick $(t,p_j(a))\in q_j(U^\dagger)$, thus
  $(t,a)\in U^\dagger$.

By induction on the closure process we can assume w.l.o.g.\ that $(t,a)$
arises from $(t_1,a)\in U$ by a single mumbling or stuttering step
or that $(t,a_1)\in U$ for some $a_1\geq a$ or else that $(t,a_i)\in
U$ where $\sup_i a_i=a$. 

In the former two cases fix a strategy for the original element of
$U$. We will use this strategy to build a new one demonstrating that $(t,a)\in U'$, hence $(t,p_j(a))\in q_j(U')$ as required.

If $(t,a)$ arises by stuttering, so $t=u(\h,\h)v$ and $t_1=uv$ we play
the strategy until $u$ is worked off. If the opponent then produces a
heap $\heap'$ to match $\heap$ we answer $\heap'$.

\smallskip 

Now $[\eff_1](\h,\h',\h,\h')$ is always true (Lemma~\ref{tillem}) so this is
a legal move. Thereafter, we continue just as in the original
strategy. In the special case where $v$ is empty, we must also show
that $[\eff_3](\h_1,\h_1',\h,\h')$ knowing $[\eff_3](\h_1,\h_1',\k_n,\k_n')$ 
where $u=(\h_1,\k_1)\dots(\h_n,\k_n)$ and $u'=(\h_1',\k_1')\dots(\h_n',\k_n')$ is the matching trace. We have $[\eff_2](\k_n,\k_n',\h,\h')$ for otherwise opponent's playing $\h'$ would have been illegal. Since, by assumption $\eff_2\subseteq\eff_3$, we can conclude $[\eff_3](\k_n,\k_n',\h,\h')$ and then $[\eff_3](\h_1,\h_1',\h,\h')$ by Lemma~\ref{tillem}(\ref{tilmon}\&\ref{tiltrans}). 

\smallskip 

If $(t,a)$ arises by mumbling then we must have $t=u(\h_1,\h_3)v$ 
and $t_1 = u(\h_1,\h_2)(\h_2,\h_3)v$. We play until the
strategy has produced a match $\h_2'$ for $\h_2$. So far, the play has
produced a trace $u'$ matching $u$, and a state $\h_1'$ so that
$[\eff_1](\h_1,\h_1',\h_2,\h_2')$. Now, we can ask what the original
strategy would produce if we gave it (temporarily assuming opponent's
role) the state $\h_2'$ as a match for $\h_2$. Note that this is legal
because $[\eff_2](\h_2,\h_2',\h_2,\h_2')$. The strategy will then
produce $\h_3'$ such that $[\eff_1](\h_2,\h_2',\h_3,\h_3')$ and our
answer in the play on the new trace against the challenge $\h_1'$ will
be this very $\h_3'$. Indeed, by composing tiles (Lemma~\ref{tillem}) we
have $[\eff_1](\h_1,\h_1',\h_3,\h_3')$ as required. Thereafter, the
play continues according to the original strategy.

\smallskip 

For down-closure, we play the strategy against $(t,a_1)$ yielding a
match $(t',a_1')\in U'$ where $a_1 E a_1'$. That same strategy also
wins against $(t,a)$ because $a E a_1'$ since $E$ is a value
specification.
\smallskip 

For closure under [Sup], finally, pick $i$ so that $a_i\geq p_j(a)$ recalling that $a=\sup_i a_i$. Since we have a winning strategy for $(t,a_i)$, we also have one (by down-closure which was already proved) for $(t,p_j(a))$ as required. 

\medskip

Ad \ref{drei}. Suppose $aEa'$. By \ref{zwei} which we have just proved we only need to match elements of the form $((\heap,\heap)a)$. The opponent plays $\h'$ where $\rloc{\h}{\h'}{\rdsin{\eff_3}}$ and we answer with $\h'$ itself and $a'$. This is always a legal move (Lemma~\ref{tillem}) and  $aEa'$, so we win the game. 

\medskip 
 
Ad \ref{vier}.  Again, we only need to match traces of the form $((\h,\h_1),a)$ where $c(\h)=(\h_1,a)$. In this case, suppose that the opponent plays $\h'$ where $\rloc{\h}{\h'}{\eff_3}$. The assumption gives $(\h_1',a')$ such that $c'(\h')=(\h_1',a')$ and $[\eff_1](\h,\h',\h_1,\h_1')$ and $aEa'$. We thus play $\h_1'$ and $a'$ and indeed $[\eff_{1/3}](\h,\h',\h_1,\h_1')$ and $aEa'$ hold so this is a winning move. 

Ad \ref{fuenf}. Suppose 
$(f,f')\in E_1{\rightarrow} T(E_2,\eff_1,\eff_2,\eff_3)$ and $(U,U')\in T(E_1,\eff_1,\eff_2,\eff_3)$. Suppose that $(uv,b)\in ap(\textit{f},{U})$ where $(u,a)\in U$ and $(v,b)$ in $f(a)$ (note that we can ignore the $\dagger$-closure). We need to produce a trace $(u'v',b') \in ap(\textit{f'},{U'})$ such that $(u',a')\in U'$ and $(v',b')$ in $f'(a')$ and $b E_2 b'$. Assume that:
\[
  u = (\h_1,\k_1) \cdots (\h_n,\k_n) \textrm{ and } v = (\h_{n+1},\k_{n+1}) \cdots (\h_{n+m},\k_{n+m})
\]
We are given a heap $\h_1'$, such that $\rloc{\h_1}{\h_1'}{\rds(\eff_3)}$. We can use the strategy $S_1$ from $(U,U')\in T(E_1,\eff_1,\eff_2,\eff_3)$ for $(u,a)$. We play according to $S_1$ to work off the $u$-part. This results in a matching trace $u' \in U'$:
\[
    u' = (\h_1',\k_1') \cdots (\h_n',\k_n')
\]
where $[\eff_3](\h_1,\h_1',\k_n,\k_n')$ and $(a,a')\in E_2$. 
We get $(f(a),f(a'))\in T(E_2,\eff_1,\eff_2,\eff_3)$. Now, we are given a heap $\h_{n+1}'$ that is an environment move forming the tile $[\eff_2](\k_n,\k_n',\h_{n+1}\h_{n+1}')$. From the fact that $\eff_2 \subseteq \eff_3$ and Lemma~\ref{tillem}(\ref{tilrd}) we can conclude $\rloc{\h_{n+1}}{\h_{n+1}'}{\rds(\eff_3)}$. 


Thus we can continue our play by using the strategy $S_2$ from
$(f(a),f(a'))\in T(E_2,\eff_1,\eff_2,\eff_3)$ which yields a
continuation $v'$ of our trace and a final answer $b'$. It is then
clear that $(u'v',b')\in \textit{bnd}(f',U')$ so this combination of
strategies does indeed win.


\medskip 

Ad \ref{sechs}. Suppose that $(U_1,U_1')\in
T(E_1,\eff_1,\eff\cup\eff_2,\eff\cup\eff_2\cup\eff')$ and $(U_2,U_2')\in T(E_2,\eff_2,\eff\cup\eff_1,\eff\cup\eff_1\cup\eff')$ and let $(t,(a,b))\in U_1\semparallel U_2$,
thus $\textit{inter}(t_1,t_2,t)$ (ignoring $\dagger$ by item
\ref{zwei}) where $(t_1,a)\in U_1$ and $(t_2,b)\in U_2$. Let $S_1$,
$S_2$ be corresponding winning strategies.  The idea is to use $S_1$
when we are in $t_1$ and to use $S_2$ when we are in $t_2$. Supposing
that $t$ starts with a $t_1$ fragment we begin by playing according to $S_1$. Let $t$ be of the form:
\[
\begin{array}{ll}
  t = & (\h_1,\k_1) \cdots (\h_n,\k_n) (\h_{n+1},\k_{n+1}) \cdots (\h_{n+m},\k_{n+m})\\ & (\h_{n+m+1},\k_{n+m+1}) \cdots (\h_{n+m+k},\k_{n+m+k})  \cdots (\h_p,\k_p)  
\end{array}
\]
composed of pieces of the traces $t_1$ and $t_2$. Assume w.l.o.g. that the first piece $(\h_1,\k_1) \cdots (\h_n,\k_n)$ is a part of $t_1$. We are given a initial heap $\h_1'$ such that $\rloc{\h}{\h'}{\rds(\eff\cup\eff'\cup(\eff_1\sqcup\eff_2))}$. Since $\rds(\eff_1\sqcup\eff_2)=\rds(\eff_1)\cup\rds(\eff_2)$, we can apply strategy $S_1$ to guide us through the first part of the game, obtaining:
\[
  (\h_1',\k_1') \cdots (\h_n',\k_n')
\]
Moreover, we have an environment move which forms the tile $[\eff](\k_n,\k_n',\h_{n+1},\h_{n'+1})$. Thus, we have the tile $[\eff \cup \eff_1](\h_1,\h_1',\h_{n+1},\h_{n+1}')$ which can be seen as an environment move for $t_2$. Therefore, we can use strategy $S_2$ for the $U'$ and continue the game, obtaining the trace piece:
\[
  (\h_{n+1}',\k_{n+1}') \cdots (\h_{n+m}',\k_{n+m}')
\]
Now, we can return to the $S_1$ game as the trace above is seen as an environment move for $U$. Alternating these strategies, we get a trace $t$ which is in $(U \semparallel U')$. Let $(a',b')$ be the final values reached at the end. It is clear that $[\eff\cup\eff'\cup\eff_1\cup\eff_2](\h,\h',\h_p,\h_p')$ and also $aE_1a'$ and $bE_2b'$. 

It remains to assert the stronger statement  $[\eff\cup\eff'\cup(\eff_1\sqcup\eff_2)](\h,\h',\h_p,\h_p')$. To see this suppose that $\wEff\loc\in\eff_1\setminus \eff_2\setminus\eff\setminus\eff'$. 
Since the entire game can be viewed as an instance of the game $U_1$ vs $U_1'$ with interventions by $U_2$ vs.\ $U_2'$ regarded as environment interactions we have $[\eff\cup\eff_2\cup\eff'](\h,\h',\h_p,\h_p')$ so that in fact 
$\rrloc{\h}{\h_p}{\loc}$ and $\rrloc{\h'}{\h_p'}{\loc}$. The case of $\cEff{\loc}$ and $\eff_1$,$\eff_2$ interchanged is analogous.

Ad \ref{acht}. This is direct from the definition of atomic and appealing on the fact that $(U,U') \in T(E, \eff_1,\emptyset, \eff_3)$.
\end{proof}

\section{Proof of Theorem~\ref{mainzwei}}

\begin{proof}
\textbf{Commuting.}
By Theorem~\ref{main}(\ref{zwei}) we can assume our pilot trace $t$ to be 
of the form:
\[
  (\h_1,\k_1) (\h_2,\k_2) \cdots (\h_n,\k_n)\quad (\h_{n+1},\k_{n+1}) \cdots (\h_{n+m},\k_{n+m})\ (a,b)
\]
where
\[
\begin{array}{l}
  t_1  = (\h_1,\k_1) (\h_2,\k_2) \cdots (\h_n,\k_n)\ v_1 \in U_1\\
  t_2 = (\h_{n+1},\k_{n+1}) \cdots (\h_{n+m},\k_{n+m})\ v_2 \in U_2
\end{array}  
\]
We make similar use of Theorem~\ref{main}(\ref{zwei}) in the subsequent cases without explicit mention. 

We are also given a heap $\h_1'$ such that $\rloc{\h_1}{\h_1'}{\rds(\eff \cup \eff_1' \cup \eff_2')}$. Because $\eff_1' \perp \eff_2'$, $\h_1$ and $\h_{n+1}$ agree on the reads of $\eff_2'$. Thus we can start a game $U_2$ vs.\ $U_2'$ using $\h_1'$ and $t_2$. We forward all environment's moves from the main game to the side game and use the responses from the side game to answer in the main game. Suppose that the side game leads to the valid $U_2$-trace 
\[
  (\h_1',\k_1') (\h_2',\k_2') \cdots (\h_m',\k_m') \ v_2'
\]
where $v_2 E_2 v_2'$ and (1) $[\eff^C \cup \eff_2'](\h_{n+1},\h_1',\k_{n+m},\k_m')$. 
Notice that in the global game these are legal responses as $[\eff_1^C \cup \eff_2^C](\h_i,\h_i',\k_{i},\k_{i}')$ for $1 \leq i \leq m$.

We now have an environment move $[\eff](\k_m,\k_{m}',\h_{m+1},\h_{m+1}')$. Since $\eff_1' \perp \eff$ and $\eff_2' \perp \eff_1'$, the heaps $\h_{1}'$ and $\h_{m+1}'$ agree in the reads of $\eff_1'$. Therefore, we can run a game $U_1$ vs.\ $U_1'$  using $\h_{m+1}'$ and $t_1$, obtaining the trace: 
\[
  (\h_{m+1}',\k_{m+1}') (\h_{m+2}',\k_{m+2}') \cdots (\h_{m+n}',\k_{m+n}') \ v_1'
\]
where $v_1 E_1 v_1'$ and (2) $[\eff^C \cup \eff_1'](\h_{1},\h_{m+1}',\k_{n},\k_{m+n}')$. The reasoning is similar to the use of the previous game.

Thus we have that $(v_1,v_2) (E_1 \times E_2)  (v_1',v_2')$.

Now, we need to conclude that $[\eff^C \cup \eff_1' \cup \eff_2'](\h_{1},\h_1',\k_{n+m},\k_{m+n}')$. This follows from the fact that $\eff_1' \perp \eff_2'$ and (1) and (2). In particular, from (1) and $\eff_1' \perp \eff_2'$, we get that $\k_{m+n}$ and $\k_{m+n}'$ agree on the locations in $\eff_2'$, while from (2), we get that $\k_{m+n}$ and $\k_{m+n}'$ agree on the locations in $\eff_1'$. This finishes the proof.

\textbf{Duplicated.}
Assume given a trace in $U$:
\[
  t = (\h_1,\k_1) \cdots (\h_n,\k_n) \ v
\]
and a heap $\h_1'$ such that $\rloc{\h_1}{\h_1'}{\rdsin{\eff_2 \cup \eff'}}$. Since $\eff_2 \perp \eff_1$ and $\rdsin{\eff'} \cap \wrsin{\eff'} = \emptyset$, we have that $\h_1$ and $\k_n$ agree on the reads of $\eff'$.

We start by simply stuttering:
\[
 t' =  (\h_1',\h_1') (\h_2',\h_2') \cdots (\h_n',??)
\]
where $[\eff] (\k_i,\h_{i+1},\k_i',\h_{i+1}')$ for $1 \leq i \leq n+m$. Notice that for $1 \leq i \leq n-1$, we have $[\eff_1^C](\h_i,\h_i',\h_{i+1},\h_i')$. So the stuttering moves are valid responses. 

We will now play $U_1$ vs.\ $U_1'$ to construct the missing heap ``??''. We first run a game using $\h_n'$ and $t$, where the environment moves are simply stutter moves:
\[
  (\h_n',\q_1) (\q_1,\q_2) \cdots (\q_{n-1},\q_n)\ v_1'
\]
such that $v E v_1'$ and $[\eff^C \cup \eff'](\h_1,\h_n',\k_n,\q_n)$. Notice that using stuttering environment moves are valid as $[\eff^C](\k_i,\q_i,h_{i+1},\q_i)$ for $1 \leq i \leq n-1$.

Since $\h_1$ and $\k_n$ agree on the reads of $\eff'$ and $\q_n$ and $\k_n$ agree on $\rdsin{\eff'}$ from $[\eff^C \cup \eff'](\h_1,\h_n',\k_n,\q_n)$, we can run the game $U_1$ vs.\ $U_1'$ again on $\q_n$ and $t$ with stutter environment moves:
\[
  (\q_n,\q_{n+1}) (\q_{n+1},\q_{n+2}) \cdots (\q_{n+m-1},\q_{n+m})\ v_2'
\]
where $v E v_2'$ and $[\eff^C \cup \eff'](\h_1,\q_n,\k_n,\q_{n+m})$. Thus, $(v, v) (E \times E) (v_1', v_2')$. 

We now put $??:=q_{m+n}$ which leads to a valid trace due to repeated mumbling.  Finally, 
$[\eff \cup \eff_2'](\h_1, \h_1', \k_n, \q_{n+m})$ follows
from $[\eff^C \cup \eff'](\h_1,\q_n,\k_n,\q_{n+m})$ and $\eff \perp \eff'$.

\textbf{Pure.}
We start with a trace from $\textit{rtn}(v)$, for example $(\h_1, \h_1),v$ and an arbitrary heap $\h_1'$. We now consider the game involving $U$ vs.\ $U'$ on $t,v$ and $\h_1'$:
\[
\begin{array}{l}
  t = (\q_1,\k_1)(\q_2,\k_2) \cdots (\q_n,\k_n), v\\
  t'= (\h_1',\k_1') (\k_1',\k_2') \cdots (\k_{n-1}',\k_n'), v'
\end{array}
\]
We have that $v E v'$ and $[\eff_3](\q_1,\h_1',\k_n,\k_n')$. By mumbling, $(\h_1',\k_n') \in U'$. We can reply with $\k_n'$ in the main game.

\textbf{Dead.} 
Assume given a trace of the form:
\[
  (\h_1,\k_1) \cdots (\h_n,\k_n)  \ v
\]
and $\h_1'$ such that $\rloc{\h_1}{\h_1'}{\rdsin{\eff_3}}$. We now
initiate a side game $U$ vs.\ $U'$ on this trace and respond in the
main game by stuttering. Thus, we obtain traces $(\h_1', \h_1') \cdots
(\h_n',\h_n') \ ()$ in the main game and $(\h_1', \k_1') \cdots
(\h_n',\k_n') \ v'$ in the side game.

The main trace is in $\textit{rtn}(\unitval)$. The side game tells us that 
 $v = \unitval$ and that $\gloc{\h_i}{\k_i}{\eff_1}$ and therefore
$[\eff_1^C](\h_i,\h_i',\k_i,\h_i')$. It remains to
show that $[\eff \cup \eff_1'\cup \eff_2'](\h_1,\h_1',\k_n,\k_n')$. This follows from the
fact that $\eff_1$ has only reads as $\h_i$ and $\k_i$ agree on all
locations.

\paragraph{Parallelization.}




We start with a trace in $\mypar{U_1}{U_2}$. Assume that the trace is of the following form:
\[
t_{1,1} t_{2,1} t_{1,2} t_{2,2} \ldots t_{1,n} t_{2,n} \ (v_1,v_2)
\]
where each $t_{i,j}$ is a possibly empty sequence of moves of the form $(\h_{i,j}^1,\k_{i,j}^1) \cdots (\h_{i,j}^{m_{i,j}},\k_{i,j}^{m_{i,j}})$ and
\[
  \begin{array}{l}
  t_1 = t_{1,1} \cdots t_{1,n} \ v_1 \in U_1\\
  t_2 = t_{2,1} \cdots t_{2,n} \ v_2 \in U_2
  \end{array}
\]
are traces from $U_1$ and $U_2$, respectively. We are also given a heap $\h_1'$ such that $\rloc{\h_{1,1}^1}{\h_1'}{\rdsin{\eff \cup \eff_1' \cup \eff_2'}}$. We also have $\rloc{\h_{1,1}^1}{\h_1'}{\rdsin{\eff^C \cup \eff_2^C \cup \eff_1'}}$. We run a side game $U_1$ vs.\ $U_1'$ 
using $\h_1'$ and $t_1$, yielding:
\[
 t_{1,1}' \cdots t_{1,n}' \ v_1'  
\]
 Assume that $(\h_1',\k_1')$ and $(\h_o',\k_o')$ are, respectively, the first and last moves of this trace. We have $v_1 E_1 v_1'$ and (1) $[\eff^C \cup \eff_2^C \cup \eff_1'](\h_{1,1}^1,\h_1',\k_{1,n}^m,\k_o')$. Notice that these are legal moves in the global game as we have $[\eff_1^C \cup \eff_2^C]$ tiles for the player moves and $[\eff]$ times for the environment moves.

 Now, assume there is an environment move $(\k_o,\h_{o+1}')$.  Since $\eff_1 \perp \eff_2$ and $\eff \perp \eff_2$, the heaps $h_{1,1}^1$ and $h_{2,1}^1$ agree on the reads of $\eff_2'$ and $\h_1'$ and $\h_{o+1}'$ also agree on the reads of $\eff_2'$. (Notice as well that $\wrsin{\eff_1} \cap \rdsin{\eff_2'} = \emptyset$ as $\eff^C \cup \eff_1^C \cup \eff_2$ is a valid effect.) Therefore, we can invoke an $U_2$ game using $\h_{o+1}'$ and $t_2$, obtaining the trace:  
\[
 t_{2,1}' \cdots t_{2,n}' \ v_2'  
\]
Assume that $(\h_{o+1}',\k_{o+1}')$ and $(\h_{o+p}',\k_{o+p}')$ are, respectively, the first and last moves of this trace. We have $v_2 E_2 v_2'$ and (2) $[\eff^C \cup \eff_1^C \cup \eff_2'](\h_{2,1}^1,\h_{o+1}',\k_{2,n}^m,\k_{o+p}')$. For the same reasons as above, these are legal moves in the global game.

Therefore $(v_1,v_2) (E_1 \times E_2) (v_1',v_2')$. 

We need now to prove that $[\eff \cup \eff_1' \cup \eff_2'](\h_{1,1}^1,\h_1',\k_{2,n}^m,\k_{o+p})$. From (1) and $\eff_1 \perp \eff_2$ and $\eff \perp \eff_1$, we have that $\k_{2,n}^m$ and $\k_{o+p}$ agree on the locations of $\eff_1$. Similarly, $\k_{2,n}^m$ and $\k_{o+p}$ agree on the locations of $\eff_2$. Since there are only $\eff$ tiles and $\eff \perp \eff_1$ and $\eff \perp \eff_2$, $\k_{2,n}^m$ and $\k_{o+p}$ agree on the locations of $\eff$. This finishes the proof.

\end{proof}

\end{document}